\documentclass[11pt]{article}
\usepackage[margin=1in]{geometry}

\usepackage{amsmath,amssymb,amsthm} \usepackage{hhline}
\usepackage{bm}
\usepackage[noline, ruled,linesnumbered,commentsnumbered]{algorithm2e}

\newtheorem{theorem}{Theorem}
\newtheorem{lemma}[theorem]{Lemma}
\newtheorem{proposition}[theorem]{Proposition}
\newtheorem{corollary}[theorem]{Corollary}
\newtheorem*{theorem*}{Theorem} \newtheorem*{lemma*}{Lemma}

\theoremstyle{definition}
\newtheorem{definition}[theorem]{Definition}
\newtheorem*{remark*}{Remark}

\usepackage{todonotes} 
\usepackage[colorlinks=true]{hyperref}

\newcommand{\abs}[1]{\left\vert#1\right\vert}
\newcommand{\set}[1]{\left\{#1\right\}}
\newcommand{\tuple}[1]{\left(#1\right)} \newcommand{\eps}{\varepsilon}
 
\newcommand{\defeq}{:=}

 \renewcommand{\mid}{\;\middle\vert\;}

\usepackage{xifthen}

\renewcommand{\Pr}[2][]{ \ifthenelse{\isempty{#1}}
  {\mathbf{Pr}\left[#2\right]} {\mathbf{Pr}_{#1}\left[#2\right]} }
\newcommand{\E}[2][]{ \ifthenelse{\isempty{#1}}
  {\mathop{\mathbf{E}}\left[#2\right]}
  {\mathop{\mathbf{E}}_{#1}\left[#2\right]} }
\renewcommand{\deg}[2][]{\mathrm{deg}_{#1}\tuple{#2}}
\newcommand{\dist}[2][]{\mathrm{dist}_{#1}\tuple{#2}}

\newcommand{\Tsaw}[2]{T_{\mathrm{SAW}}\tuple{#1,#2}}
\renewcommand{\exp}[1]{\mathrm{exp}\tuple{#1}}
\newcommand{\ER}{Erd\H{o}s-R\'{e}nyi}
\newcommand{\marg}{\mathtt{marg}}

\begin{document}

\title{Spatial mixing and approximate counting for Potts model on
  graphs with bounded average degree}

\author{ Yitong Yin\thanks{Supported by NSFC grants 61272081 and
    61321491. }\\Nanjing University, China\\ \texttt{yinyt@nju.edu.cn}
  \and
  Chihao Zhang\\
  Shanghai Jiao Tong University, China\\
  \texttt{chihao.zhang@gmail.com}}


\date{}\maketitle
\begin{abstract}
  We propose a notion of contraction function for a family of graphs
  and establish its connection to the strong spatial mixing for spin
  systems.  More specifically, we show that for anti-ferromagnetic
  Potts model on families of graphs characterized by a specific
  contraction function, the model exhibits strong spatial mixing, and
  if further the graphs exhibit certain local sparsity which are very
  natural and easy to satisfy by typical sparse graphs, then we also
  have FPTAS for computing the partition function.

  This new characterization of strong spatial mixing of multi-spin
  system does not require maximum degree of the graphs to be bounded,
  but instead it relates the decay of correlation of the model to a
  notion of effective average degree measured by the contraction of a
  function on the family of graphs. It also generalizes other notion
  of effective average degree which may determine the strong spatial
  mixing, such as the connective constant~\cite{sinclair2013spatial,
    sinclairspatial}, whose connection to strong spatial mixing is
  only known for very simple models and is not extendable to general
  spin systems.

  As direct consequences: (1) we obtain FPTAS for the partition
  function of $q$-state anti-ferromagnetic Potts model with activity
  $0\le\beta<1$ on graphs of maximum degree bounded by $\Delta$ when
  $q> 3(1-\beta)\Delta+1$, improving the previous best bound
  $\beta> 3(1-\beta)\Delta$~\cite{lu2013improved} and asymptotically
  approaching the inapproximability threshold
  $q=(1-\beta)\Delta$~\cite{galanis2013inapproximability}; and (2) we
  obtain an efficient sampler (in the same sense of fully
  polynomial-time almost uniform sampler, FPAUS) for the Potts model
  on \ER~random graph $\mathcal{G}(n,d/n)$ with sufficiently large
  constant $d$, provided that $q> 3(1-\beta)d+4$. In particular when
  $\beta=0$, the sampler becomes an FPAUS for for proper $q$-coloring
  in $\mathcal{G}(n,d/n)$ with $q> 3d+4$, improving the current best
  bound $q> 5.5d$ for FPAUS for $q$-coloring in
  $\mathcal{G}(n,d/n)$~\cite{efthymiou2013mcmc}.

\end{abstract}

\newcommand{\SAW}{\mathsf{SAW}}
\newcommand{\err}{\mathcal{E}}

\section{Introduction}


Spin systems are idealized models for local interactions with statistical behavior. In Computer Science, spin systems are widely used as a model for counting and sampling problems.
The \emph{Potts model} is a class of spin systems parameterized by the number of \emph{spin states} $q\ge 2$ and an \emph{activity} $\beta\ge 0$. Given an undirected graph $G=(V,E)$, a \emph{configuration} is a $\sigma\in [q]^{V}$ that assigns each vertex in the graph one of the $q$ states in $[q]$. Every configuration $\sigma$ is  assigned by the model with weight $w_G(\sigma)\defeq\prod_{e=uv\in E}\beta^{\mathbf{1}(\sigma(u)=\sigma(v))}$. A probability distribution over all configurations, called the \emph{Gibbs measure}, can be naturally defined as $\mu(\sigma)=\frac{w_G(\sigma)}{Z}$ where the normalizing factor $Z=Z(G)=\sum_{\sigma\in[q]^V}w_G(\sigma)$ is the \emph{partition function} in statistical physics. When $\beta<1$ the interacting neighbors favor disagreeing spin states over agreeing ones, and the model is said to be \emph{anti-ferromagnetic}.

The partition function gives a general formulation of counting problems on graphs, whose exact computation is \#P-hard. For example when $\beta=0$, the partition function $Z(G)$ gives the number of proper $q$-colorings of graph $G$. 
There is a substantial body of works on the approximate counting of proper $q$-colorings of graphs in the context of rapid mixing of a generic random walk, called the \emph{Glauber dynamics}\cite{col_Jerrum95,BD97,vigoda2000improved,col_DF03,col_Molloy04,col_HV03,hayes2006coupling,col_Hayes03,col_DFHV04}.

An exciting accomplishment in recent years is in relating the approximability of partition function for spin systems to the phase transition of the model on the infinite regular trees, also known as Bethe lattices. Here the exact phase transition property we are concerned with is the presence of the \emph{decay of correlation}, also called the \emph{spatial mixing}: assuming arbitrary possible boundary conditions on all vertices at distance $\ell$ from the root in the regular tree, the error for the marginal distribution at the root measured by the total variation distance goes to 0 as $\ell\to\infty$. The decay of correlation of the model on the infinite $\Delta$-regular tree undergoes phase transition as the activity parameter $\beta$ crosses the critical threshold in terms of $q$ and $\Delta$, called the \emph{uniqueness threshold} as it corresponds to the transition threshold for the uniqueness of Gibbs measures on the infinite regular tree. For proper $q$-colorings, the uniqueness condition for Gibbs measure on the $d$-regular tree is $q\ge \Delta+1$~\cite{jonasson2002uniqueness}, and for anti-ferromagnetic Potts model with $0<\beta<1$, the uniqueness threshold on the $\Delta$-regular tree is conjectured to be $q=(1-\beta)d$ which was proved to be the threshold for the uniqueness of semi-translation invariant Gibbs measures on $d$-regular tree~\cite{galanis2013inapproximability}. 
For anti-ferromagnetic 2-spin systems (where $q=2$), it was settled through a series of works~\cite{Weitz06, SST, LLY13, Sly10, galanis2012inapproximability, SS12} that the transition of approximability of the partition function for the model on graphs of bounded maximum degree is precisely characterized by the phase transition of the model on regular trees.

For multi-spin systems (where $q\ge 3$), in a seminal work of Gamarnik and Katz~\cite{GK07}, the decay of correlation was used to give deterministic FPTAS for counting proper $q$-colorings of graphs and for computing the partition function of general multi-spin systems. This is the first deterministic approximation algorithm for multi-spin systems and also one of the first few deterministic approximation algorithms for \#P-hard counting problems.
The specific notion of decay of correlation established in~\cite{GK07} is a stronger one, namely the \emph{strong spatial mixing}, where the correlation decay is required to hold even conditioning on an arbitrary configuration partially specified on a subset of vertices. 
Later in~\cite{gamarnik2013strong}, the strong spatial mixing was established for proper $q$-colorings on graphs of maximum degree bounded by $\Delta$ when $q\ge \alpha \Delta+1$ for $\alpha>\alpha^*\approx1.763..$, which is best bound known for strong spatial mixing for colorings of graphs of bounded maximum degree.
When the  parameters of the model are in the nonuniqueness regime, there is long-range correlation. In~\cite{galanis2013inapproximability}, this was used to establish the inapproximability of the partition function for anti-ferromagnetic Potts model on graphs with maximum degree beyond the uniqueness threshold.

In this paper, we are interested in the spatial mixing and FPTAS for spin systems on graphs with unbounded maximum degree. For some special 2-spin systems such as the hardcore model and Ising model with zero field, this was achieved by relating the decay of correlation property to a notion of effective degree of a family of graphs, called \emph{connective constant}. Roughly speaking, the connective constant of a family $\mathcal{G}$ of graphs is bounded by $\Delta$ if for all graphs $G$ from $\mathcal{G}$ the number of self-avoiding walks in $G$ of length $\ell$ starting from any vertex $v$ is bounded by $n^{O(1)}\Delta^\ell$. In~\cite{sinclair2013spatial, sinclairspatial}, the exact reliance on the maximum degree by the correlation decay and FPTAS for the hardcore model in~\cite{Weitz06} was replaced by that on the connective constant. However, as pointed out in~\cite{sinclairspatial}, the approach does not extend to general 2-spin systems. 
For multi-spin systems, the cases for unbounded maximum degree was studied in the literature mostly in the context of sampling proper $q$-colorings of \ER~graphs $\mathcal{G}(n,d/n)$ with constant average degree $d$~\cite{col_DFFV06, efthymiou2008random, mossel2010gibbs, efthymiou2013mcmc, efthymiou2012simple, efthymiou2014switching}.

Our goals are to establish strong spatial mixing and to give FPTAS for multi-spin systems on families of general graphs with unbounded maximum degree that are not restricted to $\mathcal{G}(n,d/n)$.
It turns out that for multi-spin systems, in order to achieve these goals we need a more robust way than connective constant to measure the effective average degree. 
We illustrate the necessity of this by a class of bad instances of graphs called \emph{caterpillars}, which were also considered in~\cite{yin2014spatial} and~\cite{efthymiou2014switching}. The caterpillars as in Figure~\ref{fig:caterpillar} are paths $P=(v_1,v_2,\dots,v_n)$ with each $v_i$ adjoined with $k$ bristles. Consider in particular the proper $q$-colorings. For the caterpillars with $k\ge q-2$, the vertices on bristles can always be fixed in a way to force the vertices on the path to alternate between two colors, which means the strong spatial mixing does not hold on any caterpillar with a $k\ge q-2$. On the other hand, the connective constant of the family of all caterpillars is only 1. This simple example shows how different multi-spin systems can be from the Ising model on graphs of unbounded maximum degree. It also suggests that for multi-spin systems we should take into account both the branching factor (which measures the long-range growth of self-avoiding walks) and the actual degrees of individual vertices (which locally affect the decay of correlation) when studying the decay of correlation on graphs of unbounded maximum degree.

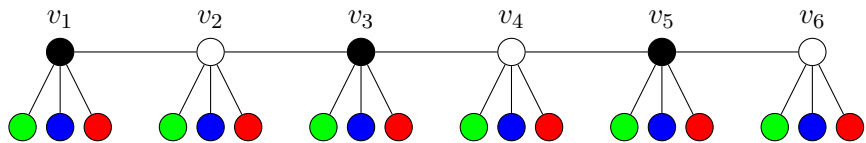
\begin{figure}[h]
  \centering
  \begin{tikzpicture}

\tikzstyle{onode}=[draw,circle,scale=1]{}
\node [onode,fill=white,label=above:{$v_4$}] (v14) at (-3,0) {};
\node [onode,fill=black,label=above:{$v_3$}] (v6) at (-5,0) {};
\node [onode,fill=black,label=above:{$v_5$}] (v18) at (-1,0) {};
\node [onode,fill=white,label=above:{$v_6$}] (v22) at (1,0) {};
\node [onode,fill=white,label=above:{$v_2$}] (v5) at (-7,0) {};
\node [onode,fill=black,label=above:{$v_1$}] (v2) at (-9,0) {};
\node [onode,fill=green] (v1) at (-9.5,-1) {};
\node [onode,fill=red] (v4) at (-8.5,-1) {};
\node [onode,fill=blue] (v3) at (-9,-1) {};
\node [onode,fill=green] (v7) at (-7.5,-1) {};
\node [onode,fill=red] (v8) at (-6.5,-1) {};
\node [onode,fill=blue] (v9) at (-7,-1) {};
\node [onode,fill=green] (v10) at (-5.5,-1) {};
\node [onode,fill=red] (v11) at (-4.5,-1) {};
\node [onode,fill=blue] (v12) at (-5,-1) {};
\node [onode,fill=green] (v13) at (-3.5,-1) {};
\node [onode,fill=red] (v15) at (-2.5,-1) {};
\node [onode,fill=blue] (v16) at (-3,-1) {};
\node [onode,fill=green] (v17) at (-1.5,-1) {};
\node [onode,fill=red] (v19) at (-0.5,-1) {};
\node [onode,fill=blue] (v20) at (-1,-1) {};
\node [onode,fill=green] (v21) at (0.5,-1) {};
\node [onode,fill=red] (v23) at (1.5,-1) {};
\node [onode,fill=blue] (v24) at (1,-1) {};
\draw  (v1) edge (v2);
\draw  (v2) edge (v3);
\draw  (v2) edge (v4);
\draw  (v2) edge (v5);
\draw  (v5) edge (v6);
\draw  (v7) edge (v5);
\draw  (v8) edge (v5);
\draw  (v5) edge (v9);
\draw  (v10) edge (v6);
\draw  (v11) edge (v6);
\draw  (v6) edge (v12);
\draw  (v13) edge (v14);
\draw  (v15) edge (v14);
\draw  (v14) edge (v16);
\draw  (v17) edge (v18);
\draw  (v19) edge (v18);
\draw  (v18) edge (v20);
\draw  (v21) edge (v22);
\draw  (v23) edge (v22);
\draw  (v22) edge (v24);
\draw  (v6) edge (v14);
\draw  (v14) edge (v18);
\draw  (v18) edge (v22);
\end{tikzpicture}
  \caption{A caterpillar with $n=6$, $k=3$ and $q=5$.}
  \label{fig:caterpillar}
\end{figure}


\subsection{Contributions}

In this paper, we prove strong spatial mixing and give FPTAS for
$q$-state anti-ferromagnetic Potts model and in particular the
$q$-colorings, for families of sparse graphs with unbounded maximum
degree. We achieve this by relating strong spatial mixing on a family
of graphs to the function which is \emph{contracting} on the graphs
from the family, a notion that generalizes the connective
constant~\cite{sinclair2013spatial, sinclairspatial} and properly
measures the effective average degree that affects the decay of
correlation in general spin systems.

As connective constant, it is convenient to talk about contraction
function on infinite graphs.  Given a vertex $v$ in a locally finite
infinite graph $G(V,E)$, let $\SAW(v,\ell)$ denote the set of
self-avoiding walks in $G$ of length $\ell$ starting at $v$. The
function $\delta:\mathbb{N}\to\mathbb{R}^+$ is a \emph{contraction
  function} for graph $G$ if
\begin{align*}
  \sup_{v\in V}\limsup_{\ell\to\infty}\err_\delta(v,\ell)=0
  \qquad\text{where}\quad
  \err_\delta(v,\ell)
  :=
  \sum_{\substack{(v,v_i,\ldots,v_\ell) \\ \in\SAW(v,\ell)}}
  \prod_{i=1}^\ell \delta(\deg{v_i}),
\end{align*}
This definition can be naturally extended to a family $\mathcal{G}$ of
finite graphs, in such a way that $\delta(\cdot)$ is contracting for
$\mathcal{G}$ if $\err_\delta(v,\ell)$ is of exponential decay in
$\ell$ for all $G\in \mathcal{G}$ and any vertex $v$ in $G$ (see
Section~\ref{section-prelim}).


We can use the contraction function $\delta(\cdot)$ to describe
various families $\mathcal{G}$ of graphs.  For example, the families
$\mathcal{G}$ of graphs with maximum degree bounded by $\Delta$ can be
described precisely by the contraction function such that
$\delta(d)=\frac{1}{\Delta}$ if $d\le \Delta$ and $\delta(d)=\infty$
if otherwise.  The contraction function also gives a more robust way
than the connective constant to capture average degrees.

\begin{proposition}
  The families of graphs $\mathcal{G}$ with connective constant
  bounded strictly by $\Delta$ are precisely the families
  $\mathcal{G}$ for which the constant function
  $\delta(d)=\frac{1}{\Delta}$ is a contraction function.
\end{proposition}

The FPTAS in~\cite{sinclairspatial} for the hardcore model with
activity $\lambda$ on the families of graphs with bounded connective
constant is actually an FPTAS for any family $\mathcal{G}$ of graphs
for which a critically defined $\delta(\cdot)$ is a contraction
function.\footnote{This $\delta(\cdot)$ is defined formally as
  follows:
  $\delta(d)=d^{\frac{1}{\rho}-1}\left(\frac{dx-1}{x+1}\right)^{\frac{1}{\rho}}$,
  where $x$ is the unique positive solution to
  $dx=\lambda(1+x)^{-d}+1$ and
  $\rho=2-(\Delta_c-1)\ln\left(1+\frac{1}{\Delta_c-1}\right)$ where
  $\Delta_c=\Delta_c(\lambda)$ is the critical (real) degree
  satisfying
  $\lambda=\frac{\Delta_c^{\Delta_c}}{(\Delta_c-1)^{\Delta_c+1}}$.}
It was proved in~\cite{sinclairspatial} that
$\delta(d)\le\frac{1}{\Delta_c}$ for all $d$ with the equality holds
precisely at the critical threshold $d=\Delta_c$, therefore with the
notion of contraction function we observe that the FPTAS
in~\cite{sinclairspatial} works on strictly broader families of graphs
than what was guaranteed in~\cite{sinclairspatial} with the connective
constant.



Our first result relates the decay of correlation (in the sense of
strong spatial mixing) of the anti-ferromagnetic Potts model on a
family $\mathcal{G}$ of general graphs to its contraction function.

\begin{theorem}[Main theorem: strong spatial mixing]\label{main-thm-ssm}
  Let $q\ge3$ be an integer and $0\le\beta< 1$. Let $\mathcal{G}$ be a
  family of finite graphs that satisfy the followings:
  \begin{itemize}
  \item the following $\delta(\cdot)$ is a contraction function for
    $\mathcal{G}$:
    \begin{align}
      \delta(d)=
      \begin{cases}
        \frac{2(1-\beta)}{q-1-(1-\beta)d}&\mbox{if }d\le\frac{q-1}{1-\beta}-2,\\
        1 & \mbox{otherwise;}
      \end{cases}\label{eq:delta-function}
    \end{align}
  \item (proper $q$-coloring) if $\beta =0$, the family $\mathcal{G}$
    also needs to be $q$-colorable.
  \end{itemize}
  Then the $q$-state Potts model with activity $\beta$ exhibits strong
  spatial mixing on all graphs in $\mathcal{G}$.
\end{theorem}

\begin{remark*}
  The contraction function describes a notion of average degree and
  the theorem holds for graphs with unbounded maximum degree.  For
  example, for \ER{} random graph $\mathcal{G}(n,d/n)$, which with
  high probability has constant average degree $(1-o(1))d$ and
  unbounded maximum degree $\Theta(\log n/\log\log n)$, assuming
  $q>3(1-\beta)d+O(1)$ with high probability the above $\delta(\cdot)$
  is a contraction function for $\mathcal{G}(n,d/n)$. Note that
  $q=(1-\beta_c)\Delta$ is the uniqueness/nonuniqueness threshold for
  semi-translation invariant Gibbs measures of Potts model on the
  infinite $\Delta$-regular tree $\mathbb{T}_\Delta$, which is also
  conjectured to be the uniqueness/nonuniqueness
  threshold~\cite{galanis2013inapproximability}. \end{remark*}

In order for the SSM to imply an FPTAS for computing the partition
function, we need the graphs to be sparse in a slightly more
restrictive manner than what guaranteed by the above contraction
function. The conditions of~\eqref{eq:delta-function} classify the
vertices in a graph into low-degree vertices (if
$\deg{v}<\frac{q-1}{1-\beta}-2$) and high-degree vertices (if
otherwise). A graph $G=(V,E)$ is said to be \emph{locally sparse} if
for every path $P$ in $G$ of length $\ell$, the total size of clusters
of high-degree vertices growing from path $P$ is bounded by
$O(\ell+\log |V|)$ (see Section~\ref{section-prelim} for a formal
definition).  Intuitively, a graph is locally sparse if all clusters
of high-degree vertices are small (of size $O(\log |V|)$) and are
relatively far away from each other.

\begin{theorem}[Main theorem: approximate counting]\label{main-thm-FPTAS}
  Let $q\ge3$ be an integer and $0\le\beta< 1$. Let $\mathcal{G}$ be a
  family of locally sparse graphs satisfying the conditions of
  Theorem~\ref{main-thm-ssm}. Then there is an FPTAS for the partition
  function of the $q$-state Potts model with activity $\beta$ for all
  graphs in $\mathcal{G}$
\end{theorem}

\begin{remark*}
  This is the first FPTAS for the general Potts model on families of
  graphs of unbounded maximum degree.  And even for graphs with
  bounded maximum degree $\Delta$ Theorem~\ref{main-thm-FPTAS} holds
  with a much better bound $q> 3(1-\beta)\Delta+1$, which greatly
  improves the best previous bound $3(1-\beta)\Delta<\beta$ known for
  anti-ferromagnetic Potts model on graphs with bounded maximum
  degree~\cite{lu2013improved}. Recall that $q=(1-\beta_c)\Delta$ is
  the semi-translation invariant uniqueness threshold on the infinite
  $d$-regular tree $\mathbb{T}_\Delta$, and the problem is hard to
  approximate for all even $q$ satisfying
  $q<(1-\beta)\Delta$~\cite{galanis2013inapproximability}.
\end{remark*}

To evaluate how tight Theorem~\ref{main-thm-ssm}
and~\ref{main-thm-FPTAS} could be, we instead consider an idealized
goal: to find a $\delta(\cdot)$ (depending only on the model) such
that the strong spatial mixing rate with respect to any $v$ in a graph
$G$ is always bounded by $\err_\delta(v,\ell)$.  With this stronger
requirement, for $q$-coloring, the following $\delta(\cdot)$ is the
best we can hope for:
\begin{align*}
  \delta(d)=
  \begin{cases}
    \frac{1}{q-d-1}&\mbox{if }d\le q-2\\
    1 & \mbox{otherwise}
  \end{cases}
\end{align*}
because it is achieved by the SSM rate on caterpillars as illustrated
in Figure~\ref{fig:caterpillar} when the colors of leaves are fixed
properly.  This certainly gives a lower bound to any contraction
function for the above idealized goal.  Note that the $\delta(\cdot)$
function in~\eqref{eq:delta-function} in the case of $q$-coloring
(when $\beta=0$) is precisely twice this lower bound function. And
this factor 2 is due to an intrinsic obstacle in the current
approaches for correlation-decay based algorithms for multi-spin
systems.





As an application of Theorem~\ref{main-thm-FPTAS} we consider the
\ER{} random graph $\mathcal{G}(n,d/n)$ with constant average degree
$d$. It is well known that the partition function of the Potts model
on this random graph is highly concentrated to its expectation which
is easy to calculate~\cite{coja2013chasing}, so for this model
sampling is more interesting than counting. And because the FPTAS in
Theorem~\ref{main-thm-FPTAS} actually works for a broader
self-reducible family of instances (for example, the list-colorings
instead of just $q$-colorings), we also have efficient sampling
algorithms due to the standard Jerrum-Valiant-Vazirani
reduction~\cite{samp_JVV86}.
So we prove the following result.

\begin{theorem}\label{main-thm-random}
  Let $d$ be sufficiently large and $G\sim\mathcal{G}(n,d/n)$. Let
  $0\le\beta<1$ and $q> 3(1-\beta)d+4$.  There exists an algorithm
  $\mathcal{S}$ such that for any $\epsilon>0$ with high probability
  $\mathcal{S}$ returns a random configuration in $[q]^{V(G)}$ from a
  distribution that is within total variation distance $\epsilon$ from
  the Gibbs distribution $\mu_{G}$ for the $q$-state Potts model with
  activity $\beta$.  And the running time of $\mathcal{S}$ is in
  polynomial in $n$ and $\log\frac{1}{\epsilon}$.  When $\beta=0$,
  i.e.~for $q>3d+4$, with high probability $\mathcal{S}$ is an FPAUS
  (fully polynomial-time almost uniform sampler) for proper
  $q$-colorings of $G$.
\end{theorem}

This is the first result for sampling in general Potts model on
$\mathcal{G}(n,d/n)$ and also improves the state of the arts for the
special case of sampling proper $q$-colorings of $\mathcal{G}(n,d/n)$.
For this special case of the problem, FPAUSes have been obtained
mostly by the rapidly mixing of certain block Glauber
dynamics~\cite{col_DFFV06, efthymiou2008random, mossel2010gibbs,
  efthymiou2013mcmc}.  The best bound $q>5.5 d$ was achieved
in~\cite{efthymiou2013mcmc}, whereas our bound for proper $q$-coloring
is $q>3d+4$.  Further, the results of~\cite{col_DFFV06,
  efthymiou2008random, mossel2010gibbs, efthymiou2013mcmc} are more
restricted to $\mathcal{G}(n,d/n)$. In contrast, our algorithm is more
generic and works for all families of locally sparse graphs with a
bounded growth rate of self-avoiding walks measured properly by the
contraction function.


In~\cite{efthymiou2012simple,efthymiou2014switching}, sampling
algorithms with a much weaker control of total variation errors than
FPAUS were considered, and with this weaker sampler a better bound
$q\ge (1+\eps)d$ was achieved, which almost approaches the uniqueness
threshold.

\subsection{Techniques}


The approximation algorithms for multi-spin systems (e.g.~graph
coloring) have been studied in the literature in the context of rapid
mixing of Glauber dynamics.  In a seminal work of Gamarnik and
Katz~\cite{GK07}, the correlation-decay based deterministic algorithms
for multi-spin systems are introduced. Our analysis of decay of
correlation utilizes the approach of~\cite{GK07} in an essential
manner.

Consider for example the proper $q$-colorings of graph $G$. If there
is a vertex $v$ in $G$ with degree much higher than $q$, then the
Glauber dynamics will have torpid mixing around $v$ since the color of
$v$ will be frozen at most of the time; and the bound on the decay of
correlation also breaks at this high-degree vertex because locally it
may have absolute correlation with the neighbors. Further, for graph
coloring with maximum degree unbounded, even the feasibility becomes
an issue.  These bad situations were dealt with in~\cite{col_DFFV06,
  mossel2010gibbs, efthymiou2013mcmc} by using block dynamics, where
the block contains the high-degree vertices as its core which is
separated from the boundary by a buffer of small-degree vertices.  As
the bound on $q$ getting tighter, the constructions of such blocks
have to be highly delicate to meet this requirement.

The novelty in our approach is that we give a correlation-decay based
deterministic algorithm that works in terms of blocks. A key
observation of us is that despite the correlation can be absolute
between a pair of high-degree vertices within the same block, its
contribution to the decay rate along a self-avoiding walk is at most a
factor 1 (hence the $\delta(d)=1$ branch for large $d$
in~\eqref{eq:delta-function}), while the low-degree vertices at the
boundaries of blocks contribute to the decay of correlation as in the
bounded degree case. While this observation is made to the original
Gibbs measure, algorithmically it can be witnessed by applying the
recursion in terms of marginal distributions on blocks.  In contrast
to the block construction in~\cite{col_DFFV06, efthymiou2008random,
  mossel2010gibbs, efthymiou2013mcmc}, the blocks in our algorithm are
extremely simple and generic: they are just clusters of high-degree
vertices.  This simple construction of blocks makes our algorithm more
generic and works on general graphs.

In a previous work~\cite{yin2014spatial}, this idea of block version
of decay of correlation was used to establish a ``spatial mixing
only'' result for random graphs, augmented from the result of
Gamarnik~\emph{et al.}~\cite{gamarnik2013strong} for graphs of bounded
maximum degree. The current work emphasizes the strong spatial mixings
that have algorithmic implications. Here the property of being locally
sparse is used to bound the running time.

\newcommand{\mono}{\mathsf{\# mon}}

\section{Preliminaries}\label{section-prelim}
Let $G=(V,E)$ be an undirected graph. For any subset $S\subseteq V$ of vertice, let $G[S]$ denote the subgraph of $G$ induced by $S$, and let $\partial B=\set{u\in V\setminus B\mid \exists w\in B,(u,w)\in E}$ denote the vertex boundary of $B$. Given a vertex $v$ in $G$, let $\dist[G]{v,S}$ denote the minimum distance from $v$ to any vertex $u\in S$ in $G$.

\paragraph{Potts model and spatial mixing.} The Potts model is parameterized by an integer $q\ge 2$ and a real $\beta\ge 0$ called the \emph{activity} parameter.  Each element of $[q]$ is called a \emph{color} or a \emph{state}. 
Let $G=(V,E)$ be a graph.
A \emph{configuration} $\sigma\in[q]^\Lambda$ on a subset $\Lambda\subseteq V$ of vertices assigns each vertex $v$ in $\Lambda$ one of the $q$ colors in $[q]$. 
In the Potts model on graph $G$, each configuration $\sigma\in[q]^V$ is assigned a weight
\[
w_G(\sigma)
=\beta^{\mono(\sigma)},
\]
where $\mono(\sigma)=|\set{(u,v)\in E \mid \sigma(u)=\sigma(v)}|$ gives the number of monochromatic (undirected) edges in the configuration $\sigma$. 

In order to study strong spatial mixing, we consider the instances of Potts model with boundary conditions. 
An instance of Potts model is a tuple $\Omega=\tuple{G,\Lambda,\sigma}$ where $G=(V,E)$ is an undirected graph, $\Lambda\subseteq V$ is a subset of vertices in $G$ and $\sigma\in[q]^\Lambda$ is a configuration on $\Lambda$.  
Given such an instance $\Omega=\tuple{G,\Lambda,\sigma}$, the weight function $w_\Omega$ assigns each configuration $\pi\in[q]^V$ the weight $w_\Omega(\pi)=w_G(\pi)$ if $\pi$ agrees with $\sigma$ over all vertices in $\Lambda$, and $w_\Omega(\pi)=0$ if otherwise.
An instance $\Omega$ is \emph{feasible} if there exists a configuration on $V$ with positive weight.  This gives rise to a nature probability distribution over all configurations on $V$ for a feasible Potts instance:
\[
\Pr[\Omega]{c(V)=\pi}=\frac{w_\Omega(\pi)}{Z(\Omega)},
\]
where $Z(\Omega)= \sum_{\sigma\in[q]^V}w_\Omega(\sigma)$ is call the \emph{partition function}. This probability distribution is called the \emph{Gibbs measure}. 
For a vertex $v\in V$ and any color $x\in[q]$, we use $\Pr[\Omega]{c(v)=x}$ to denote the marginal probability that $v$ is assigned color $x$ by a configuration sampled from the Gibbs measure. 
Similarly, for a set $S\subseteq V$ and $\pi\in [q]^S$, we use $\Pr[\Omega]{c(S)=\pi}$ to denote the marginal probability that $S$ is assigned configuration $\pi$ by a configuration sampled from the Gibbs measure. Next we define the notion of strong spatial mixing.

\begin{definition}[strong spatial mixing]
Let $\mathcal{G}$ be a family of graphs. We say that the $q$-state Potts model with activity $\beta$ exhibits \emph{strong spatial mixing} on all graphs in $\mathcal{G}$ if there exist positive constants $C>0, \gamma<1$ such that for every $G=(V,E)\in\mathcal{G}$,  every $v\in V$, $\Lambda\subseteq V$, $x\in[q]$, and any configurations $\sigma,\tau\in[q]^\Lambda$ that $\Omega_1=(G,\Lambda,\sigma)$ and $\Omega_2=(G,\Lambda,\tau)$ are both feasible, it holds that
\[
\left| \Pr[\Omega_1]{c(v)=x}-\Pr[\Omega_2]{c(v)=x}\right|\le n^C\gamma^{\dist[G]{v,\Delta}},
\]
where $n=|V|$ and $\Delta\subseteq\Lambda$ is the set of vertices on which $\sigma$ and $\tau$ differ. 
\end{definition}
If we change $\dist[G]{v,\Delta}$ to $\dist[G]{v,\Lambda}$ in the definition it becomes the definition of \emph{weak spatial mixing}.

\paragraph{Permissive block and locally sparse.}
Fix any $q\ge 2$ and $0\le\beta<1$. Let $\Omega=(G,\Lambda,\sigma)$ be an instance of $q$-state Potts model with activity $\beta$ and $v$ a vertex in $G$. 
We call $v$ a \emph{low-degree} vertex if $\mathrm{deg}_G(u)<\frac{q-1}{1-\beta}-2$, and otherwise we call it a \emph{high-degree} vertex.

\begin{definition}[permissive block]
Let $\Omega=(G,\Lambda,\sigma)$ be a Potts instance where $G=(V,E)$.
A vertex set $B\subseteq V\setminus \Lambda$ is a \emph{permissive block} in $\Omega$ if every boundary vertex $u\in\partial B\setminus\Lambda$ is a low-degree vertex. 
For any subset of vertices $S\subseteq V\setminus \Lambda$, we denote $B(S)=B_{\Omega}(S)$ the \emph{minimal permissive block} containing $S$. And we write $B(v)=B(S)$ if $S=\{v\}$ is a singleton.
\end{definition}

\begin{definition}
  A family $\mathcal{G}$ of finite graphs is \emph{locally sparse} if
  there exists a constant $C>0$ such that for every $G=(V,E)$ in the
  family and every path $P$ in $G$ of length $\ell$ we have
  $|B(P)|\le C(\ell+\log |V|)$.
\end{definition}

\paragraph{Feasibility and local feasibility.}

Let $\Omega=(G,\Lambda,\sigma)$ where $G=(V,E)$ be an instance of
Potts model, $v\in V\setminus\Lambda$ be a vertex. For a subset of
vertices $S\subseteq V\setminus\Lambda$, a configuration $\pi\in[q]^S$
is (globally) \emph{feasible} if there exists a configuration on $V$
with positive weight and agrees with $\pi$ on $S$. A configuration
$\pi\in[q]^S$ is \emph{locally feasible}, if
$w_{G[\Lambda\cup S]}(\sigma\cup\pi)>0$, where $\sigma\cup\pi$ is the
configuration over $\Lambda\cup S$ that agrees with both $\sigma$ and
$\pi$.

The disucssion of feasibility and local feasibility is meaningful only
when $\beta=0$. In this case, the local feasibility of a configuration
on a permissive block implies the (global) feasibility.

\begin{proposition}\label{prop:local-feasibility}
  Let $\Omega=(G,\Lambda,\sigma)$ where $G=(V,E)$ be an feasible
  instance of Potts model with $\beta=0$, $v\in V\setminus\Lambda$ be
  a vertex and $\pi\in [q]^{B(v)}$ be a locally feasible
  configuration. Then $\pi$ is also feasible.
\end{proposition}
\begin{proof}
  Denote $B=B(v)$. Fix a configuration $\eta\in[q]^V$ such that $w_\Omega(\eta)>0$,
  this is possible since $\Omega$ is feasible. We denote by $\eta'$
  the restriction of $\eta$ to
  $V\setminus\tuple{(B\cup\partial B)\setminus\Lambda}$, i.e., the set
  of vertices that are either in $\Lambda$, or not in
  $B\cup\partial B$.

  Consider the configuration
  $\eta=\pi\cup \eta'\in[q]^{V\setminus\tuple{\partial B\setminus
      \Lambda}}$,
  it can be extended to a configuration $\rho\in[q]^{V}$ with
  $w_{\Omega}(\rho)>0$ in a greedy fashion, since every vertex in
  $\partial B\setminus\Lambda$ is of low-degree. Thus $\rho$ witness
  that $\pi$ is feasible.
\end{proof}

With this proposition, we do not distinguish between local feasibility
and feasiblity of configurations on permissive blocks. For a
permissive block $B$, we use $\mathcal{F}(B)$ to denote the set of
feasible configuration. Note that when $\beta>0$, the set
$\mathcal{F}(B)$ is simply $[q]^B$.

\paragraph{Self-avoiding walk tree.} 
Given a graph $G=(V,E)$ and a vertex $v\in V$,  a rooted tree $T$ can be naturally constructed from all self-avoiding walks starting from $v$ in $G$ as follows: Each vertex in $T$ corresponds to a self-avoiding walk (simple path in $G$) $P=(v,v_1,v_2,\dots,v_k)$ starting from $v$, whose children correspond to all self-avoiding walks $(v,v_1,v_2,\dots,v_k,v_{k+1})$ in $G$ extending $P$, and the root of $T$ corresponds to the trivial walk $(v)$. The resulting tree, denoted by $\Tsaw{G}{v}$, is called the \emph{self-avoiding walk} (SAW) tree constructed from vertex $v$ in graph $G$.

From this construction, every vertex in $\Tsaw{G}{v}$ can be naturally identified with the vertex in $V$ (many-to-one) at which the corresponding self-avoiding walk ends. 

\paragraph{Connective constant and contraction function.}
Given a vertex $v$ in a locally finite graph $G(V,E)$, let $\SAW(v,\ell)$ denote the set of self-avoiding walks in $G$ of length $\ell$ starting at $v$. The following notion of connective constant of families of finite graphs is introduced in~\cite{sinclair2013spatial}.
\begin{definition}[connective constant~\cite{sinclair2013spatial,sinclairspatial}]
Let $\mathcal{G}$ be a family of finite graphs. 
The \emph{connective constant} of $\mathcal{G}$ is bounded by $\Delta$ if there exists a positive constant $C>0$ such that for any graph $G=(V,E)$ in $\mathcal{G}$ and any vertex $v$ in $G$, we have $|\SAW(v,\ell)|\le n^C\Delta^\ell$ where $n=|V|$ for all $\ell\ge 1$.
\end{definition}

Let $\delta:\mathbb{N}\to\mathbb{R}^+$ be a function.
Given a vertex $v$ in a locally finite graph $G(V,E)$, let
\[
\err_\delta(v,\ell)
:=
\sum_{\substack{(v,v_i,\ldots,v_\ell) \\ \in\SAW(v,\ell)}}
\prod_{i=1}^\ell \delta(\deg{v_i}).
\]

\begin{definition}[contraction function]
%
Let $\mathcal{G}$ be a family of finite graphs. The $\delta:\mathbb{N}\to\mathbb{R}^+$ is a \emph{contraction function} for $\mathcal{G}$ if there exist positive constants $C>0,\gamma<1$ such that for any graph $G=(V,E)$ in $\mathcal{G}$ and any vertex $v$ in $G$, we have $\err_\delta(v,\ell)<n^C\gamma^\ell$ where $n=|V|$ for all $\ell\ge 1$.
\end{definition}
It is easy to see that graph families $\mathcal{G}$ with constant contraction function $\delta(d)=\frac{1}{\Delta}$ are precisely the families $\mathcal{G}$ of connective constant  bounded strictly by $\Delta$.

\section{Recursion}\label{sec:recursion}

In this section, we introduce recursions to compute the marginal
probability on a vertex and on a permissive block in Potts model
respectively.

Let $\Omega=(G,\Lambda,\sigma)$ where $G=(V,E)$ be an instance of
Potts model and $v\in V\setminus\Lambda$ be a vertex. Let $B=B(v)$ be
the minimal permissive block containing $v$. Let
$\delta B=\set{u_iv_i\mid i\in[m]}$ be an enumeration of boundary
edges of $B$ where $v_i\not\in B$ for every $i\in[m]$.  In this
notation, more than one $u_i$ or $v_i$ may refer to the same
vertex. We denote $E(B)\defeq\set{uv\in E\mid u,v\in B}$ the edges in
$B$.  We use $\bar B$ to denote the inner boundary of $B$, i.e.,
$\bar B=\set{u\in B\mid uv\in E\mbox{ and }v\not\in B}$.


Recall that we use $\mathcal{F}(B)$ to denote the set of feasible
configurations on a permissive block $B$, it is easy to see that, for
every $x\in[q]$,
\[
  \Pr[\Omega]{c(v)=x}=\sum_{\substack{\pi\in
      \mathcal{F}(B)\\\pi(v)=x}}\Pr[\Omega]{c(B)=\pi}.
\]
This identity relates the marginal probability on a vertex to marginal
probabilities on a block. We now define notations for some
sub-instances and give a block-to-vertices identity.

Let $\pi\in \mathcal{F}(B)$ be a configuration on a permissive block
$B$. For every $i\in[m]$, denote $\pi_i=\pi(u_i)$.  Let
$G_B=(V_B,E_B)$ denote the graph obtained from $G$ by removing
$B\setminus\bar B$ and edges in $E(B)$, i.e.,
$V'=(V\setminus B)\cup \bar B$, $E'=E\setminus E(B)$. Let
$\Omega_B=(G_B,\Lambda,\sigma)$. For every $i=1,2,\dots,m+1$, define
$\Omega^\pi_i=(G_i^\pi,\Lambda_i^\pi,\sigma_i^\pi)$ as the instance
obtained from $\Omega_B$ by fixing $u_j$ to color $\pi_j$ for every
$j\in[i-1]$ and by removing edges $u_jv_j$ for every
$j=i,i+1,\dots,m$.

\begin{lemma}\label{lem:block-recursion}
  Assuming above notations, it holds that
  \begin{equation}
    \label{eq:block-vertex-recursion}
    \Pr[\Omega]{c(B)=\pi}
    =\frac{w_{G[B]}(\pi)\cdot\prod_{i=1}^m\tuple{1-(1-\beta)\Pr[\Omega_i^\pi]{c(v_i)=\pi_i}}}
    { \sum_{\rho\in
        \mathcal{F}(B)}w_{G[B]}(\rho)\cdot\prod_{i=1}^m\tuple{1-(1-\beta)\Pr[\Omega_i^\rho]{c(v_i)=\rho_i}}
    }.
  \end{equation}

\end{lemma}

\begin{proof}
  \begin{align*}
    \Pr[\Omega]{c(B)=\pi}
      &=\frac{w_{G[B]}(\pi)\cdot Z(\Omega_{m+1}^\pi)}{\sum_{\rho\in\mathcal{F}(B)}w_{G[B]}(\rho)\cdot Z(\Omega_{m+1}^\rho)}\\
      &=\frac{w_{G[B]}(\pi)\cdot \frac{Z(\Omega_{m+1}^\pi)}{Z(\Omega_1^\pi)}}{\sum_{\rho\in\mathcal{F}(B)}w_{G[B]}(\rho)\cdot \frac{Z(\Omega_{m+1}^\rho)}{Z(\Omega_1^\rho)}}\\
      &=\frac{w_{G[B]}(\pi)\cdot \prod_{i=1}^m\frac{Z(\Omega_{i+1}^\pi)}{Z(\Omega_i^\pi)}}{\sum_{\rho\in\mathcal{F}(B)}w_{G[B]}(\rho)\cdot \prod_{i=1}^m\frac{Z(\Omega_{i+1}^\rho)}{Z(\Omega_i^\rho)}}.
  \end{align*}
  Since for every $\rho\in\mathcal{F}(B)$ and $i\in[d]$,
  \[
    Z(\Omega_{i+1}^\rho)=\sum_{y\in[q]}Z\tuple{\Omega_{i}^\rho\mid
      c(v_i)=y}\cdot \beta^{\mathbf{1}(y=\rho(u_i))}
  \]
  where $Z\tuple{\Omega_{i}^\rho\mid c(v_i)=y}$ stands for the sum of
  the weights of all feasible configurations $\sigma$ on
  $\Omega_i^\rho$ satisfying $\sigma(v_i)=y$ and $\mathbf{1}(\cdot)$
  is the indicator function.  With this identity, we can further write
  \begin{align*}
    \Pr[\Omega]{c(B)=\pi}
    &=\frac{w_{G[B]}(\pi)\cdot \prod_{i=1}^m\frac{\sum_{y\in[q]}Z\tuple{\Omega_{i}^\pi\mid c(v_i)=y}\cdot
      \beta^{\mathbf{1}(y=\pi(u_i))}}{Z(\Omega_i^\pi)}}{\sum_{\rho\in\mathcal{F}(B)}w_{G[B]}(\rho)\cdot \prod_{i=1}^m\frac{\sum_{y\in[q]}Z\tuple{\Omega_{i}^\rho\mid c(v_i)=y}\cdot
      \beta^{\mathbf{1}(y=\rho(u_i))}}{Z(\Omega_i^\rho)}}\\
    &=\frac{w_{G[B]}(\pi)\cdot\prod_{i=1}^m\tuple{1-(1-\beta)\Pr[\Omega_i^\pi]{c(v_i)=\pi_i}}}
      { \sum_{\rho\in
      \mathcal{F}(B)}w_{G[B]}(\rho)\cdot\prod_{i=1}^m\tuple{1-(1-\beta)\Pr[\Omega_i^\rho]{c(v_i)=\rho_i}}
      }.
  \end{align*}
\end{proof}

This identity expresses the marginal probability on a permissive block
as the function of marginal probabilities on its incident vertices,
with modified instances. We now analyze the derivatives of this
function.

\begin{lemma}\label{lem:tech}
  Let
  $
  \mathbf{p}=\tuple{p_{i,\rho}}_{i\in[m],\rho\in\mathcal{F}(B)},\mathbf{\hat
    p}=\tuple{\hat p_{i,\rho}}_{i\in[m],\rho\in\mathcal{F}(B)}$
  be two tuples of variables and
  \[
    f(\mathbf{p})\defeq
    \frac{w_{G[B]}(\pi)\prod_{i=1}^m\tuple{1-(1-\beta)p_{i,\pi}}}
    {\sum_{\rho\in\mathcal{F}(B)}w_{G[B]}(\rho)\prod_{i=1}^m\tuple{1-(1-\beta)p_{i,\rho}}}.
  \]
  Assume for every $i\in[m],\rho\in \mathcal{F}(B(v))$,
  $p_{i,\rho},\hat p_{i,\rho}\le \frac{1-\beta}{q-(1-\beta)d_i}$, then
  \[
    \abs{\log f(\mathbf{p})-\log f(\mathbf{\hat p})}\le
    \sum_{i\in[d]}\frac{2(1-\beta)}{q-(1-\beta)d_i-1}\cdot\max_{\rho\in[q]^{B(v)}}\abs{\log
      p_{i,\rho}-\log\hat p_{i,\rho}}.
  \]
\end{lemma}
\begin{proof}
  For every $i\in[m]$, we have
  \[
    \frac{\partial f}{\partial
      p_{i,\pi}}=-(1-\beta)f(1-f)\cdot\frac{1}{1-(1-\beta)p_{i,\pi}}.
  \]
  For every $i\in[m]$ and $\rho\ne \pi$, we have
  \[
    \frac{\partial f}{\partial p_{i,\rho}}
    =(1-\beta)f\cdot\frac{w_{G[B]}(\rho)\prod_{i=1}^m(1-(1-\beta)p_{i,\rho})}{\sum_{\sigma\in\mathcal{F}(B)}w_{G[B]}(\sigma)\prod_{i=1}^m(1-(1-\beta)p_{i,\sigma})}\cdot\frac{1}{1-(1-\beta)p_{i,\rho}}.
  \]
  Thus,
  \[
    \sum_{\substack{\rho\in\mathcal{F}(B)\\\rho\ne\pi}}\frac{\partial
      f}{\partial p_{i,\rho}}\le (1-\beta)f(1-f)\cdot
    \max_{\substack{\rho\in\mathcal{F}(B)\\\rho\ne
        \pi}}\frac{1}{1-(1-\beta)p_{i,\rho}}.
  \]
  Let $\Phi=\frac{1}{x}$, by mean value theorem, for some
  $\mathbf{\tilde p}=(\tilde p_{i,\rho})_{i\in[m],\rho\in[q]^{B(v)}}$
  where each $\tilde p_{i,\rho}\le \frac{1-\beta}{q-(1-\beta)d_i}$, we
  have
  \begin{align*}
    &\quad\;\abs{\log f(\mathbf{p})-\log f(\mathbf{\hat p})}\\
    &=\sum_{i\in[m]}\sum_{\rho\in\mathcal{F}(B)}\left.\tuple{\frac{\Phi(f)}{\Phi(p_{i,\rho})}\abs{\frac{\partial
      f}{\partial p_{i,\rho}}}}\right|_{\mathbf{p}=\mathbf{\tilde p}}\cdot \abs{\log p_{i,\rho}-\log\hat p_{i,\rho}}\\
    &\le\sum_{i\in[m]}\left.\tuple{\frac{\Phi(f)}{\Phi(p_{i,\pi})}\abs{\frac{\partial
      f}{\partial p_{i,\pi}}}+\sum_{\substack{\rho\in\mathcal{F}(B)\\\rho\ne \pi}}\frac{\Phi(f)}{\Phi(p_{i,\rho})}\abs{\frac{\partial
    f}{\partial p_{i,\rho}}}}\right|_{\mathbf{p}=\mathbf{\tilde p}}\cdot\max_{\rho\in[q]^{B(v)}}\abs{\log p_{i,\rho}-\log\hat p_{i,\rho}}\\
    &\le \sum_{i\in[m]}\left.\tuple{(1-\beta)\tuple{\frac{p_{i,\pi}}{1-(1-\beta)p_{i,\pi}}+\max_{\substack{\rho\in\mathcal{F}(B)\\\rho\ne\pi}}\frac{p_{i,\rho}}{1-(1-\beta)p_{i,\rho}}}}\right|_{\mathbf{p}=\mathbf{\tilde p}}\cdot\max_{\rho\in[q]^{B(v)}}\abs{\log p_{i,\rho}-\log\hat p_{i,\rho}}\\
    &\le \sum_{i\in[m]}\frac{2(1-\beta)}{q-(1-\beta)d_i-(1-\beta)}\cdot \max_{\rho\in[q]^{B(v)}}\abs{\log p_{i,\rho}-\log\hat p_{i,\rho}}.
  \end{align*}  

\end{proof}

The following lemma gives an upper bound for the probability
$\Pr[\Omega]{c(v)=x}$.

\begin{lemma}\label{lem:marg-upperbound}
  Assume $q> (1-\beta) d$. For every color $x\in [q]$, it holds that
  \[
    \Pr[\Omega]{c(v)=x}\le \frac{1}{q-(1-\beta) d}.
  \]
  where $d$ is the degree of $v$ in $G$.
\end{lemma}
\begin{proof}
  Assume $x=1$.  For every $i\in [q]$, let $x_i$ denote the number of
  neighbors of $v$ that are of color $i$. Then
  $p_{v,1}\le \max \frac{\beta^{x_1}}{\sum_{i\in[q]}\beta^{x_i}}$
  subject to the constraints that all $x_i$ are nonnegative integers
  and $\sum_{i=1}^{q}x_i=d$. Since $\beta\le 1$, we can assume
  $x_1=0$, thus
  $p_{v,1}\le \max \frac{1}{1+\sum_{i=2}^{q}\beta^{x_i}}$. We now
  distinguish between two cases:
  \begin{enumerate}
  \item (If $d\ge q-1$) In this case, let $\lambda=1-\beta$, then
    \[
      \frac{1}{1+\sum_{i=2}^q\beta^{x_i}}\le
      \frac{1}{1+(q-1)(1-\lambda)^{\frac{d}{q-1}}}
      \overset{\heartsuit}{\le} \frac{1}{1+(q-1)\tuple{1-\frac{\lambda
            d}{q-1}}} =\frac{1}{q-(1-\beta) d},
    \]
    where $\heartsuit$ is due to the fact that the inequality
    $(1-a)^b\ge 1-ab$ holds when $0\le a\le 1$ and $b\ge 1$.
  \item (If $d< q-1$) In this case, due to the integral constraint of
    $x_i$'s, the term $\sum_{i=2}^q\beta^{x_i}$ minimizes when $d$ of
    $x_i$'s are set to one and remaining $x_i$'s are set to
    zero. Therefore, we have
    \[
      \frac{1}{1+\sum_{i=2}^q\beta^{x_i}}\le
      \frac{1}{1+d\beta+(q-1-d)} =\frac{1}{q-(1-\beta) d},
    \]
  \end{enumerate}
\end{proof}

The recursion \eqref{eq:block-vertex-recursion} holds for arbitrary
set of vertices $B$ (not necessary a permissive block), thus if one
takes $B$ as a single vertex, it implies the following simple lower
bound for marginal probabilities on a vertex.

\begin{lemma}\label{lem:marg-lowerbound}
  For every feasible $x\in[q]$, it holds that
  \[
    \Pr[\Omega]{c(v)=x}\ge \frac{\beta^d}{q},
  \]
  where $d$ is the degree of $v$ in $G$.
\end{lemma}


\section{Strong Spatial Mixing}\label{sec:ssm}

We prove the strong spatial mixing property for Potts model in this
section. Recall that
\[
  \delta(d)=
  \begin{cases}
    \frac{2(1-\beta)}{q-1-(1-\beta)d} &\mbox{if }d\le\frac{q-1}{1-\beta}-2\\
    1 &\mbox{otherwise.}
  \end{cases}
\]
Theorem \ref{main-thm-ssm} is restated in a formal way:

\begin{theorem}\label{thm:ssm}
  Let $q\ge 3$ be an integer and $0\le\beta< 1$. Let $\mathcal{G}$ be
  a family of finite graphs that satisfy the followings:
  \begin{itemize}
  \item the function $\delta(\cdot)$ is a contraction function for
    $\mathcal{G}$;
  \item (proper $q$-coloring) if $\beta =0$, then $\mathcal{G}$ is a
    family of $q$-colorable graphs.
  \end{itemize}
  Then there exist two constants $C_1,C_2>0$ such that the following
  holds: For every graph $G(V,E)\in \mathcal{G}$ with $\abs{V}=n$,
  every vertex $v\in V$, every color $x\in[q]$, every set of vertices
  $\Lambda\subseteq V\setminus\set{v}$ and two feasible instances
  $\Omega_1=(G,\Lambda,\rho)$, $\Omega_2=(G,\Lambda,\pi)$ with
  $\rho,\pi\in [q]^\Lambda$ being two configurations on $\Lambda$, it
  holds that
  \[
    \abs{\Pr[\Omega_1]{c(v)=x}-\Pr[\Omega_2]{c(v)=x}}\le
    n^{C_1}\cdot\exp{-C_2\cdot\ell},
  \]
  where $\ell\defeq\mathrm{dist}(v,\Delta)$ and
  $\Delta\subseteq \Lambda$ is the subset of $\Lambda$ on which $\rho$
  and $\pi$ differ.
\end{theorem}

We prove the theorem by using the recursion introduced in Section
\ref{sec:recursion} to estimate marginal probability on a vertex
$v$. In each step, we show that the difference between the (logarithm
of) marginal probabilities caused by different configurations on
$\Delta$ contracts by a factor of $\delta(\cdot)$, and therefore
relate the difference of marginal probabilities to the contraction
function.

The following observation is useful: If $\delta(\cdot)$ is a
contraction function for $\mathcal{G}$, then for every graph
$G\in\mathcal{G}$, every sufficiently long path in $G$ must contain a
low-degree vertex. The property is formally stated as:

\begin{lemma}\label{lem:cut-exist}
  Let $\mathcal{G}$ be a family of finite graphs for which
  $\delta(\cdot)$ is a contraction function.  Then for some constants
  $\theta>1$ and $C>0$, for every $G=(V,E)\in\mathcal{G}$ with
  $\abs{V}=n$, every $v\in V$ and every $L\ge C\log n$, there exists a
  low-degree $S$ in $T=\Tsaw{G}{v}$ such that for every $u\in S$,
  $L< \dist[T]{u,v}\le \theta L$ and every self-avoiding walk in $T$
  from $v$ of length $\theta L$ intersects $S$.
\end{lemma}
\begin{proof}
  Let $G(V,E)\in\mathcal{G}$ be a graph. It follows from the
  definition of contraction function that for some constant $C>0$, for
  every $\ell\ge C\log n$, $\err_\delta(v,\ell)<\alpha^\ell$ for some
  constant $0<\alpha<1$.
  
  It is sufficient to show that, for some constant integer $\theta>0$
  it holds that for every $v\in V$, every $L\ge C\log n$, every
  $P=(v,v_1,\dots,v_{\theta L})\in\SAW(v,\theta L)$, there exists a
  low-degree vertex $v_j$ among
  $\set{v_{L+1},v_{\theta L},\dots,v_{\theta L}}$.
  
  Let
  $\theta=\max\set{\lceil
    \log_{1/\alpha}\tuple{\frac{q-1}{2(1-\beta)}}\rceil,2}$.
  Assume for the contradiction that every vertex in\\
  $\set{v_{L+1},v_{L+2},\dots,v_{\theta L}}$ has high-degree. Since
  $\theta L> L\ge C\log n$, we have
  $\prod_{i=1}^{\theta L}\delta(\deg{v_i})\le \alpha^{\theta L}$.

  On the other hand, since
  $\delta(d)\ge\delta(0)=\frac{2(1-\beta)}{q-1}$, we have
  $\prod_{i=1}^{\theta L}\delta(\deg{v_i})\ge
  \tuple{\frac{2(1-\beta)}{q-1}}^{L}$.
  This is a contradiction for our choice of $\theta$.
\end{proof}

\subsection{The $\beta>0$ case}

To implement the recursion introduced in Section \ref{sec:recursion},
we define two procedures $\mathtt{marg}(\Omega,v,x,\ell)$ and
$\mathtt{marg\mbox{-}block}(\Omega,B(v),\pi,\ell)$ calling each other
to estimate vertex and block marginal respectively.  We assume
$\Omega=(G,\Lambda,\sigma)$ where $G=(V,E)$ is a feasible instance of
Potts model, $v\in V\setminus \Lambda$ is a vertex, $x\in[q]$ is a
color and $\ell$ is an integer. Recall that for a permissive block
$B(v)$, we use $\mathcal{F}(B)$ to denote the set of feasible
configurations over $B(v)$.

\begin{algorithm}[H]\label{algo:marg-prob-ssm}
  \caption{$\mathtt{marg}(\Omega,v,x,\ell)$}
  \label{algo:marg-prob-ssm:fixed}
  If $v$ is fixed to be color $y$, then return $1$ if $x=y$ and return
  $0$ if $x\ne y$\;%
  If $\ell<0$ return $1/q$\;%
  \label{algo:marg-prob-ssm:ell}
  Compute $B(v)$\;%
  For every $\rho\in \mathcal{F}(B(v))$, let
  $\hat p_\rho\gets
  \mathtt{marg\mbox{-}block}(\Omega,B(v),\rho,\ell)$\;
  Return
  $\min\set{\sum_{\substack{\pi\in
        \mathcal{F}(B(v))\\s.t.~\pi(v)=x}}\hat{p}_{\pi},\frac{1}{\max\set{1,q-(1-\beta)\deg[G]{v}}}}$
  \label{algo:marg-prob-ssm:return}
\end{algorithm}

To describe the algorithm for estimating the block marginals, we need
to introduce some notations. Let $B=B(v)$, and we enumerate the
boundary edges in $\delta B$ by $e_i=u_iv_i$ for $i=1,2,\ldots, m$,
where $v_i\not\in B$. With this notation more than one $u_i$ or $v_i$
may refer to the same vertex, which is fine. For every $i\in[m]$ and
$\rho\in \mathcal{F}(B)$, define $\Omega_B$ and $\Omega_i^\rho$ as in
Lemma \ref{lem:block-recursion}.

Let $P_i=(v, w_1, w_2,\ldots, w_k, v_i)$ be a self-avoiding walk from
$v$ to $v_i$ such that all intermediate vertices $w_i$ are in $B(v)$.
Since $B(v)$ is a minimal permissive block, such walk always exists,
and let $P_i$ be an arbitrary one of them if there are multiple ones.



\begin{algorithm}[H]\label{algo:marg-prob-block-ssm}
  \caption{$\mathtt{marg\mbox{-}block}(\Omega,B(v),\pi,\ell)$}
  Compute $P_i$ for every $i\in[m]$\;
  $\hat
  p_{i,\rho}\gets\mathtt{marg}(\Omega_i^\rho,v_i,\rho_i,\ell-\abs{P_i})$
  for every $i\in[m]$ and $\rho\in \mathcal{F}(B)$\;%
  Return
  $ \frac{w_{G[B]}(\pi) \prod_{i\in[m]}\tuple{1-(1-\beta)\hat
      p_{i,\pi}}}
  {\sum_{\rho\in\mathcal{F}(B)}w_{G[B]}(\rho)\prod_{i\in[m]}\tuple{1-(1-\beta)\hat
      p_{i,\rho}}}$\;
\end{algorithm}

We need a few definitions to analyze the two procedures.

\begin{definition}\label{def:computation-tree}
  Given an instance $\Omega=(G,\Lambda,\sigma)$ of Potts model where
  $G=(V,E)$, a vertex $v\in V\setminus \Lambda$, a color $x\in[q]$ and
  an integer $\ell$. The computation tree of $\marg(\Omega,v,x,\ell)$,
  denoted by $\mathcal{CT}(\Omega,v,x,\ell)$, is a rooted tree
  recursively defined as follows:
  \begin{itemize}
  \item The root of $\mathcal{CT}(\Omega,v,x,\ell)$ is labeled
    $(\Omega,v,x,\ell)$;
  \item For every recursive call to $(\Omega',v',x',\ell')$ by
    $\marg(\Omega,v,x,\ell)$ (in the subroutine
    $\mathtt{marg\mbox{-}block}$), $(\Omega,v,x,\ell)$ has a children
    which is the computation tree of $(\Omega',v',x',\ell')$.
  \end{itemize}
  Define the \emph{termination set} of $\marg(\Omega,v,x,\ell)$ as the
  set of vertices $u$ in the self-avoiding walk tree
  $\Tsaw{G[V\setminus\Lambda]}{v}$ that $\marg(\Omega',u,x',\ell')$
  returns at step \ref{algo:marg-prob-ssm:ell} for some leaf
  $(\Omega',u,x',\ell')$ of $\mathcal{CT}(\Omega,v,x,\ell)$.
\end{definition}
Thus the computation of $\marg(\Omega,v,x,\ell)$ ends either at
trivial instance (including vertex with fixed color and one-vertex
graph), or at vertices in termination set.

\begin{definition}\label{def:S-error}
  Given an instance $\Omega=(G,\Lambda,\sigma)$ of Potts model with
  $q\ge 3$ and activity $0<\beta<1$ where $G=(V,E)$ with $\abs{V}=n$, a vertex
  $v\in V\setminus \Lambda$.  Let $T=\Tsaw{G[V\setminus\Lambda]}{v}$
  be the self-avoiding walk tree rooted at $v$ in
  $G[V\setminus\Lambda]$ and $S$ be a set of low-degree vertices in
  $T$. Assume $v$ has $m$ children $v_1,v_2,\dots,v_m$ in $T$, let
  $T_i$ denote the subtree of $T$ rooted at $T_i$. We recursively
  define the error function:
  \[
    \err_{T,S}\defeq
    \begin{cases}
      \sum_{i=1}^m\delta(\deg[G]{v_i})\cdot\err_{T_i,S}& \mbox{if }v\not\in S,\\
      q+n\log\frac{1}{\beta}.  &\mbox{otherwise.}
    \end{cases}
  \]
\end{definition}

\begin{definition}\label{def:log-error}
  Given an instance $\Omega=(G,\Lambda,\sigma)$ of Potts model where
  $G=(V,E)$, a vertex $v\in V\setminus \Lambda$, a color $x\in[q]$, an
  assignment $\pi\in\mathcal{F}(B(v))$ and an integer $\ell$. We
  denote $p_{\Omega,v,\ell}(x)=\mathtt{marg}(\Omega,v,x,\ell)$ and
  $p_{\Omega,v,B(v),\ell}(\pi)=\mathtt{marg\mbox{-}block}(\Omega,B(v),\pi,\ell)$.
  Define
  \begin{align*}
    \mathcal{E}_{\Omega,\ell}(v)
    &\defeq \max_{x\in[q]}
      \abs{\log\tuple{p_{\Omega,v,\ell}(x)}-\log\tuple{\Pr[\Omega]{c(v)=x}}};\\
    \mathcal{E}_{\Omega,\ell}(B(v))
    &\defeq \max_{\pi\in\mathcal{F}(B(v))}\abs{\log\tuple{p_{\Omega,B(v),\ell}(\pi)}-\log\tuple{\Pr[\Omega]{c(B(v))=\pi}}}.
  \end{align*}
  We use the convention that $\log 0-\log 0=0$.
\end{definition}

The following key lemma relates the error functions we introduced
above.

\begin{lemma}\label{lem:err-uppbound}
  Let $\hat\Omega=(\hat G,\hat \Lambda,\hat\sigma)$ be an instance of
  Potts model with $q\ge 3$ and activity $0<\beta<1$ where $\hat G=(\hat V,\hat E)$.
  Assume $\abs{\hat V}=n$. Let $\hat v\in \hat V\setminus\hat \Lambda$
  be a vertex and $x\in[q]$ be a color. Let $L>0$ be an integer and
  $\hat S$ be the termination set of $\marg(\hat\Omega, \hat v,x,L)$.
  Denote $\hat T=\Tsaw{\hat G[\hat V\setminus\hat\Lambda]}{\hat v}$ as
  the self-avoiding walk tree rooted at $\hat v$ in
  $\hat G[\hat V\setminus\hat\Lambda]$. Then
  $\mathcal{E}_{\hat \Omega,L}(v)
  \le \mathcal{E}_{\hat T,\hat S}.$
\end{lemma}
\begin{proof}
  Let $\mathcal{CT}=\mathcal{CT}(\hat\Omega,\hat v,x,L)$. For every
  vertex $(\Omega,v,z,\ell)$ in $\mathcal{CT}$ where
  $\Omega=(G=(V,E),\Lambda,\sigma)$, we apply induction on the depth
  of $\mathcal{CT}(\Omega,v,z,\ell)$ to show
  \[
    \mathcal{E}_{\Omega,\ell}(v)\le
    \mathcal{E}_{T,\hat S\cap V(T)},
  \]
  where $T=\Tsaw{G[V\setminus\Lambda]}{v}$ and $V(T)$ is the set of
  vertices in $T$.
  

  The base case is that $\marg(\Omega,v,z,\ell)$ is itself a leaf,
  namely it returns without any further recursive call to
  $\mathtt{marg}$.  Then
  \begin{itemize}
  \item if it is returned at step \ref{algo:marg-prob-ssm:fixed},
    $\err_{\Omega,\ell}(v)=0$;
  \item if it is returned at step \ref{algo:marg-prob-ssm:return},
    $\err_{\Omega,\ell}(v)=0$;
  \item if it is returned at step \ref{algo:marg-prob-ssm:ell}, due to
    Lemma \ref{lem:marg-lowerbound},
    $\err_{\Omega,\ell}(v)\le q+\deg[G]{v}\log\frac{1}{\beta}$.
  \end{itemize}


  Assume the lemma holds for smaller depth and
  $\mathtt{marg}(\Omega,v,z,\ell)$ is not a leaf. In
  Algorithm~\ref{algo:marg-prob-ssm}, the estimation of marginal is
  computed as:
  \[
    p_{\Omega,v,\ell}(z)= \min\set{\sum_{\substack{\pi\in
          \mathcal{F}(B(v))\\s.t.~\pi(v)=z}}p_{\Omega,B(v),\ell}(\pi),\frac{1}{\max\set{1,q-(1-\beta)\deg[G]{v}}}}
  \]

  By Lemma~\ref{lem:marg-upperbound}, it always holds that
  $\Pr[\Omega]{c(v)=z}\le\frac{1}{\max\set{1,q-(1-\beta)\deg[G]{v}}}$.
  Thus assuming
  $p_{\Omega,v,\ell}(z)=\sum_{\substack{\pi\in
      \mathcal{F}(B(v))\\s.t.~\pi(v)=z}} p_{\Omega,B(v),\ell}(\pi) $
  will not make the error $\err_{T,S}(v)$ smaller, and hence we have
  \begin{align*}
    \err_{\Omega,\ell}(v)
    &=
      \max_{x\in [q]}\abs{\log\tuple{\Pr[\Omega]{c(v)=x}}
      -\log p_{\Omega,v,\ell}(x)}\\
    &=
      \max_{x\in [q]}\abs{\log\tuple{\sum_{\substack{\pi\in \mathcal{F}(B(v))\\s.t.~\pi(v)=x}}\Pr[\Omega]{c(B(v))=\pi}} -\log\tuple{\sum_{\substack{\pi\in \mathcal{F}(B(v))\\s.t.~\pi(v)=x}}p_{\Omega,B(v),\ell}(\pi)}}\\
    &\le \max_{\pi\in \mathcal{F}(B(v))}
      \abs{\log\tuple{\Pr[\Omega]{\sigma(B(v))=\pi}}- \log\tuple{p_{\Omega,B(v),\ell}(\pi)}} \\
    &= \err_{\Omega,\ell}(B(v)).
  \end{align*}
  where the last inequality is due to that for every positive
  $a_1,\dots,a_n$, $b_1,\dots,b_n$,
  $\frac{\sum_{i\in[n]}a_i}{\sum_{i\in[n]}b_i}\le\max_{i\in[n]}\frac{a_i}{b_i}$.

  Since $\mathtt{marg}(\Omega,v,z,\ell)$ is not a leaf, the value of
  $p_{\Omega,v,\ell}(z)$ is returned at step
  \ref{algo:marg-prob-ssm:return} of
  Algorithm~\ref{algo:marg-prob-ssm}, and it is computed from the
  recursion in Algorithm~\ref{algo:marg-prob-block-ssm}.  Recall
  $\delta B(v)=\set{u_iv_i\mid i\in[m]}$. For every $i\in[m]$, we let
  $\ell_i=\ell-\abs{P_i}$. We claim that Lemma \ref{lem:tech} implies
  \begin{align*}
    \mathcal{E}_{\Omega,\ell}(B(v)) 
    &\le
      \sum_{i\in[m]}\frac{2(1-\beta)}{q-(1-\beta)\deg[G]{v_i}-1}\cdot\max_{\rho\in
      \mathcal{F}(B(v))}\abs{\log\tuple{\Pr[\Omega_i^\rho]{\sigma(v_i)=\rho_i}}-\log\tuple{p_{\Omega^\rho_i,v_i,\ell_i}(\rho_i)}}\\
    &=
      \sum_{i\in[m]}\frac{2(1-\beta)}{q-\deg[G]{v_i}-1}\cdot\max_{\rho\in
      \mathcal{F}(B(v))}\mathcal{E}_{\Omega^\rho_i,\ell_i}(v_i),
  \end{align*}
  where $\Omega^\rho_i$ is obtained from $\Omega$ as in Lemma
  \ref{lem:tech}.

  To see this, note all $v_i$ is on the boundary of a permissive
  block, thus either $\err_{\Omega_i^\rho,\ell-\abs{P_i}}(v_i)=0$ for
  every $\rho$ (in case that the color of $v_i$ is fixed), or by Lemma
  \ref{lem:marg-upperbound},
  \[
    \Pr[\Omega_i^\rho]{c(v_i)=\rho_i}\le\frac{1}{q-(1-\beta)\mathrm{deg}_G(v_i)}.
  \]
  Also from step \ref{algo:marg-prob-ssm:return} of Algorithm
  \ref{algo:marg-prob-ssm}, we have
  \[
    p_{\Omega^\rho_i,v_i,\ell_i}(\rho_i)\le\frac{1}{q-(1-\beta)\mathrm{deg}_G(v_i)}.
  \]
  
  With this upper bound for $\err_{\Omega,\ell}(v)$, and note that
  every $P_i$ is a self-avoiding walk from $v$ to $v_i$ with every
  vertex in $B(v)$, we can then apply the induction hypothesis to
  complete the proof.
\end{proof}

We are now ready to prove the main theorem of this section.

\begin{proof}[Proof of Theorem \ref{thm:ssm} when $\beta>0$]
  By the definition of the strong spatial mixing, it is sufficient to
  prove the theorem for $\ell=\Omega(\log n)$. Let $C$ and $\theta$ be
  the constants in Lemma \ref{lem:cut-exist} and assume
  $\ell>\theta\lceil C\log n\rceil$ be an integer. Let
  $L=\lfloor \ell/\theta\rfloor\ge C\log n$.  Consider
  $p_{\Omega_1,v,L}(x)=\marg(\Omega_1,v,x,L)$. Let $S$ denote the
  termination set of $\marg(\Omega_1,v,x,L)$. By our choice of $L$,
  Lemma \ref{lem:cut-exist} implies that the set $S$ satisfies
  \begin{enumerate}
  \item every vertex in $S$ is of distance $(L,2L]$ to $v$ in
    $T=\Tsaw{G[V\setminus \Lambda]}{v}$, i.e.,
    $L<\dist[T]{v,u}\le \theta L$ for every $u\in S$;
  \item every path from $v$ to $\Delta$ in $T$ intersects $S$.
  \end{enumerate}
  It follows from Lemma \ref{lem:err-uppbound} that
  \[
    \abs{\log\tuple{p_{\Omega_1,v,L}(x)}-\log\tuple{\Pr[\Omega_1]{c(v)=x}}}\le
    \err_{\Omega_1,L}(v)\le \err_{T,S},
  \]
  and similarly
  \[
    \abs{\log\tuple{p_{\Omega_2,v,L}(x)}-\log\tuple{\Pr[\Omega_2]{c(v)=x}}}\le
    \err_{\Omega_2,L}(v)\le \err_{T,S}.
  \]
  On the otherhand, we have
  \[
    \err_{T,S}\le \tuple{q+\deg[G]{v}\log\frac{1}{\beta}}\cdot
    \sum_{k=L+1}^{\theta L}\err_{\delta}(v,k)\le
    n^{C_1'}\cdot\exp{-C_2'\cdot\ell}
  \]
  for some constants $C_1',C_2'>0$.
  
  Note that by the second property of $S$,
  $p_{\Omega_1,v,x,L}=p_{\Omega_2,v,x,L}$, thus
  \begin{align*}
    \abs{\Pr[\Omega_1]{c(v)=x}-\Pr[\Omega_2]{c(v)=x}}
    &\le\abs{\Pr[\Omega_1]{c(v)=x}-p_{\Omega_1,v,L}(x)}+\abs{\Pr[\Omega_2]{c(v)=x}-p_{\Omega_1,v,L}(x)}\\
    &\le n^{C_1}\cdot\exp{-C_2\cdot\ell}
  \end{align*}
  for some constants $C_1,C_2>0$.
\end{proof}



\subsection{The $\beta=0$ case (Coloring model)}

Since the lower bound for marginal probability in Lemma
\ref{lem:marg-lowerbound} is zero for $\beta=0$, the quantity
$\err_{T,S}$ defined in Definition \ref{def:S-error} is no longer
bounded above. We slightly modify the procedure $\mathtt{marg}$ to
deal with this case.

Let $\Omega=(G,\Lambda,\sigma)$ be an instance of Potts model with
$q\ge 3$ and activity $\beta=0$ where $G=(V,E)$, $v\in V\setminus\Lambda$ be a
vertex, $x\in[q]$ be a color and $\ell$ be an integer. We define

\begin{algorithm}[H]\label{algo:marg-prob-ssm-color}
  \caption{$\mathtt{marg}(\Omega,v,x,\ell)$}
  If $v$ is fixed to be color $y$, then return $1$ if $x=y$ and return
  $0$ if $x\ne y$\;%
  \label{algo:marg-prob-ssm-color:fixed}
  Compute $B(v)$\;%
  If $\ell<0$, then return $1/q$ if there is a feasible
  $\pi\in \mathcal{F}(B(v))$ such that $\pi(v)=x$ and return $0$ if no
  such $\pi$ exists\;
  \label{algo:marg-prob-ssm-color:ell}
  For every $\rho\in \mathcal{F}(B(v))$, let
  $\hat p_\rho\gets
  \mathtt{marg\mbox{-}block}(\Omega,B(v),\rho,\ell)$\;
  Return
  $\min\set{\sum_{\substack{\pi\in
        \mathcal{F}(B(v))\\s.t.~\pi(v)=x}}\hat{p}_{\pi},\frac{1}{\max\set{1,q-(1-\beta)\deg[G]{v}}}}$
  \label{algo:marg-prob-ssm-color:return}
\end{algorithm}

The only difference of this version of $\marg$ is at
step \label{algo:marg-prob-ssm-color:ell}, where we check whether the
color $x$ is locally feasible. We return $1/q$ if so and return $0$
otherwise.

Let $T=\Tsaw{G[V\setminus\Lambda]}{v}$ be the self-avoiding walk tree
rooted at $v$ in $G[V\setminus\Lambda]$. With our new version of
$\marg$, define the computation tree
$\mathcal{CT}=\mathcal{CT}(\Omega,v,x,\ell)$ the same as in Definition
\ref{def:computation-tree}, while the termination set of
$\marg(\Omega,v,x,\ell)$ is defined as the set of vertices $u$ in $T$
that $\marg(\Omega',u,x',\ell')$ returns at step
\ref{algo:marg-prob-ssm-color:ell} in Algorithm
\ref{algo:marg-prob-ssm-color} for some leaf $(\Omega',u,x',\ell')$ of
$\mathcal{CT}$.

We can similarly define error functions as the $\beta>0$ case, with
difference on the base case.

\begin{definition}\label{def:S-error-color}
  Given an instance $\Omega=(G,\Lambda,\sigma)$ of Potts model with
  $q\ge 3$ and activity $\beta=0$ where $G=(V,E)$ with $\abs{V}=n$, a vertex
  $v\in V\setminus \Lambda$.  Let $T=\Tsaw{G[V\setminus\Lambda]}{v}$
  be the self-avoiding walk tree rooted at $v$ in
  $G[V\setminus\Lambda]$ and $S$ be a set of low-degree vertices in
  $T$. Assume $v$ has $m$ children $v_1,v_2,\dots,v_m$ in $T$, let
  $T_i$ denote the subtree of $T$ rooted at $T_i$. We recursively
  define the error function:
  \[
    \err_{T,S}\defeq
    \begin{cases}
      \sum_{i=1}^m\delta(\deg[G]{v_i})\cdot\err_{T_i,S}& \mbox{if }v\not\in S,\\
      n\log q.  &\mbox{otherwise.}
    \end{cases}
  \]
\end{definition}

\begin{lemma}\label{lem:err-uppbound-color}
  Let $\hat\Omega=(\hat G,\hat \Lambda,\hat\sigma)$ be an instance of
  Potts model with $q\ge 3$ and activity $\beta=0$ where $\hat G=(\hat V,\hat E)$.
  Assume $\abs{\hat V}=n$. Let $\hat v\in \hat V\setminus\hat \Lambda$
  be a vertex and $x\in[q]$ be a color. Let $L>0$ be an integer and
  $\hat S$ be the termination set of $\marg(\hat\Omega, \hat v,x,L)$.
  Denote $\hat T=\Tsaw{\hat G[\hat V\setminus\hat\Lambda]}{\hat v}$ as
  the self-avoiding walk tree rooted hat $\hat v$ in
  $G[\hat V\setminus\hat\Lambda]$. Then
  $\mathcal{E}_{\hat \Omega,L}(v)
  \le \mathcal{E}_{\hat T,\hat S}.$
\end{lemma}
\begin{proof}
  The proof of this lemma is almost identical to the proof of Lemma
  \ref{lem:err-uppbound}, except at the base case of the induction.

  Consider the situation that $(\Omega,v,z,\ell)$ is a leaf in
  $\mathcal{CT}=\mathcal{CT}(\hat\Omega,\hat v,x,L))$, then
  \begin{itemize}
  \item if it is returned at step
    \ref{algo:marg-prob-ssm-color:fixed}, then
    $\err_{\Omega,\ell}(v)=0$;
  \item if it is returned at step
    \ref{algo:marg-prob-ssm-color:return}, then
    $\err_{\Omega,\ell}(v)=0$;
  \item if it is returned at step \ref{algo:marg-prob-ssm-color:ell},
    then we claim that $\err_{\Omega,\ell}(v)\le n\log q$. To see
    this, it is sufficient to show that if there is a feasible
    $\pi\in\mathcal{F}(B(v))$ such that $\pi(v)=x$, then
    $\Pr[\Omega]{c(v)=x}>0$, and if no such $\pi$ exists, then
    $\Pr[\Omega]{c(v)=x}=0$ (We use the convention that
    $\log 0-\log 0=0$ and the fact that $\Pr[\Omega]{c(v)=x}>0$
    implies $\Pr[\Omega]{c(v)=x}\ge q^{-n})$. This is a consequence of
    Proposition \ref{prop:local-feasibility}
  \end{itemize}
\end{proof}




\begin{proof}[Proof of Theorem \ref{thm:ssm} when $\beta=0$] 
  With Lemma \ref{lem:err-uppbound} replaced by Lemma
  \ref{lem:err-uppbound-color}, the proof is almost identical to the
  $\beta>0$ case.
\end{proof}

\section{Approximate Counting and Sampling}

In this section, we prove Theorem \ref{main-thm-FPTAS}. We first show
how to estimate the marginal probability in Potts model and it is
routine to obtain FPTAS from this estimation.

\subsection{Estimate the marginals}

\begin{theorem}\label{thm:comp-marg}

  Let $q\ge 3$ be an integer and $0\le\beta<1$. Let $\mathcal{G}$ be a
  family of finite graphs that satisfies the followings:
  \begin{itemize}
  \item the function $\delta(\cdot)$ is a contraction function for
    $\mathcal{G}$;
  \item (proper $q$-coloring) if $\beta=0$, the family $\mathcal{G}$
    is $q$-colorable;
  \item the family $\mathcal{G}$ is locally sparse.
  \end{itemize}
  Then for every feasible instance $\Omega=(G,\Lambda,\sigma)$ of
  Potts model where $G=(V,E)\in\mathcal{G}$ with $\abs{V}=n$,
  $\Lambda\subseteq V$ and $\sigma\in[q]^\Lambda$, for every vertex
  $v\in V$ and every color $x\in[q]$, there exists an algorithm that
  can compute an estimation $\hat p$ of $\Pr[\Omega]{c(v)=x}$ in time
  polynomial in $n$, satisfying
  \[
    1-O\tuple{\frac{1}{n^3}}\le \frac{ \hat p}{\Pr[\Omega]{c(v)=x}}
    \le 1+O\tuple{\frac{1}{n^3}}.
  \]
\end{theorem}

Let $\mathcal{G}$ be a family of finite graphs satisfying condition in
Theorem \ref{thm:ssm}. Let $\Omega=(G,\Lambda,\sigma)$ be an instance
of Potts model where $G=(V,E)\in\mathcal{G}$. Then for every vertex
$v\in V$, color $x\in[q]$, set of vertices
$\Delta\subseteq V\setminus\set{\Lambda\cup\set{v}}$ and a feasible
configuration $\rho\in [q]^\Delta$, we have shown in the proof of
Theorem \ref{thm:ssm} that we can compute an estimate $\hat p$ of
$\Pr[\Omega]{c(v)=x\mid c(\Delta)=\rho}$ such that, for some universal
constants $C_1,C_2>0$,
\[
  \abs{\log \hat p-\log\tuple{\Pr[\Omega]{c(v)=x\mid
        c(\Delta)=\rho}}}\le n^{C_1}\cdot\exp{-C_2\cdot\ell},
\]
where $\ell=\dist[G]{v,\Delta}$, as long as $\ell\ge C\log n$ for some
constant $C>0$.

To prove Theorem \ref{thm:comp-marg}, we show that if $\mathcal{G}$ is
locally sparse, then our estimation algorithm is also efficient, i.e.,
terminates in polynomial time for $L=O(\log \abs{V})$.


\begin{lemma}\label{lem:time-potts}
  Let $q\ge 3$ be an integer and $0\le \beta<1$. Assume $\mathcal{G}$
  is a family of graphs satisfying condition in Theorem
  \ref{thm:comp-marg}. Then there exists a constant $C>0$ such that
  for every feasible instance
  $\hat\Omega=(\hat G(\hat V,\hat E),\hat \Lambda,\hat \sigma)$ with
  $\hat G\in \mathcal{G}$, every vertex
  $\hat v\in \hat V\setminus \hat \Lambda$, every color $x\in[q]$ and
  every $L\ge C\log \abs{\hat V}$, the procedure
  $\mathtt{marg}(\hat \Omega,\hat v,x,L)$ (both $\beta>0$ and
  $\beta=0$ versions) terminates in time
  $\abs{\hat V}^{O(1)}\exp{O(L)}$.
\end{lemma}
\begin{proof}
  Let $C$ and $\theta$ be the constants in Lemma
  \ref{lem:cut-exist}. Fix $\hat S$ as the termination set of
  $\marg(\hat \Omega,\hat v,x,L)$.  Let
  $\hat T=\Tsaw{G[\hat V\setminus\hat\Lambda]}{v}$.  Then it follows
  from Lemma \ref{lem:cut-exist} that
  $L< \dist[\hat T]{v,\hat S}\le \theta L$. Let
  $\mathcal{CT}=\mathcal{CT}(\hat \Omega,\hat v,x,L)$ denote the
  computation tree of $\marg(\hat \Omega,\hat v,x,L)$.

  For every $(\Omega,v,z,\ell)\in\mathcal{CT}$, where
  $\Omega=(G(V,E),\Lambda,\sigma)$, consider the self-avoiding walk
  tree $T$ that is obtained from $\Tsaw{G[V\setminus \Lambda]}{v}$ by
  removing all descendants of $\hat S$. We use
  $\mathcal{P}_{\Omega,z}$ to denote the set of self-avoiding walks
  corresponding to the leaves of $T$.
  
  We claim that $\abs{\mathcal{P}_{\hat\Omega,\hat v}}=\exp{O(L)}$.
  To see this, note that
  $\abs{\mathcal{P}_{\hat\Omega,\hat v}}\le \sum_{k=1}^{\theta
    L}\SAW_{\hat G}(\hat v,k)$,
  where $\SAW_{\hat G}(\hat v,k)$ is the set of self-avoiding walks of
  length $k$ from $\hat v$ in $\hat G$.  On the other hand, we have
  \[
    \abs{\hat V}^{O(1)}\ge \sum_{k=1}^{\theta L}\err_{\delta}(\hat
    v,k)\ge \tuple{\frac{2(1-\beta)}{q-1}}^{\theta
      L}\sum_{k=1}^{\theta L}\SAW_{\hat G}(\hat v,k),
  \]
  where the last inequality is due to
  $\delta(d)\ge \frac{2(1-\beta)}{q-1}$ for every $d\ge 0$.
  
  The time cost of each vertex $\mathtt{marg}\tuple{\Omega,v,z,\ell}$
  in $\mathcal{CT}$ besides the recursive calls is at most
  $C'\cdot mq^{B_{\Omega}(v)}$ for some constant $C'>0$ where
  $m=\abs{\delta B_{\Omega}(v)}$ is the size of edge boundary of
  $B_\Omega(v)$.  We use $\tau_{\Omega,v}$ to denote the maximum
  running time of $\mathtt{marg}(\Omega,v,z,\ell)$ over all colors
  $z\in[q]$. We apply induction on the depth of
  $\mathcal{CT}(\marg(\Omega,v,z,\ell))$ to show that
  $\tau_{\Omega,v}\le
  C'\sum_{P\in\mathcal{P}_{\Omega,v}}q^{2\abs{\bigcup_{u\in
        P}B_\Omega(u)}}$.
  If the depth of $\mathcal{CT}(\Omega,v,z,\ell)$ is one, the upper
  bound is trivial. Now assume the lemma holds for smaller
  depth. Denote $B=B_\Omega(v)$ and assume
  $\delta B=\set{u_iv_i\mid i\in[m]}$ be the edge boundary of $B$.
  Then
  \[
    \tau_{\Omega,v}
    \le \sum_{\pi\in\mathcal{F}(B)}\sum_{i\in[m]}\tau_{\Omega_i^{\pi},v_i}+C'\cdot mq^{\abs{B}}\\
    \le q^{\abs{B}}\sum_{i\in[m]}\tuple{C'+\max_{\pi\in\mathcal{F}(B)}\tau_{\Omega_i^{\pi},v_i}}\\
  \]
  Applying the induction hypothesis, we have for some $\pi\in\mathcal{F}(B)$
  \begin{align*}
    \tau_{\Omega,v}
    &\le C'\cdot
      q^{\abs{B}}\cdot\sum_{i\in[m]}\tuple{1+\sum_{P\in
      \mathcal{P}_{\Omega_i^\pi,v_i}}q^{2\abs{\bigcup_{u\in
      P}B_{\Omega_i^{\pi}}(u)}}}\\
    &\le
      C'\sum_{i\in[m]}\sum_{P\in\mathcal{P}_{\Omega_i^\pi,v_i}}q^{2\tuple{\abs{\bigcup_{u\in
      P}B_{\Omega_i^\pi}(u)}+\abs{B}}} \\
    &\le
      C'\cdot\sum_{P\in\mathcal{P}_{\Omega,v}}q^{2\abs{\bigcup_{u\in
      P}B_{\Omega}(u)}},
  \end{align*}
  where the last inequality is due to the following three facts:
  \begin{enumerate}
  \item each path in $\mathcal{P}_{\Omega_i^\pi,v_i}$ is a part of
    some path in $\mathcal{P}_{\Omega,v}$, and all these paths in
    $\mathcal{P}_{\Omega,v}$ are distinct;
  \item $B\cap B_{\Omega_i^\pi(u)}=\varnothing$ for every $i,u$ and
    $\pi$;
  \item $B_{\Omega}(u)$ is at least as large as $B_{\Omega_i^\pi}(u)$
    for every $u$.
  \end{enumerate}
  Therefore, we have
  \[
    \tau_{\hat \Omega,\hat v}\le
    C'\cdot\sum_{P\in\mathcal{P}_{\Omega,v}}q^{2\abs{\bigcup_{u\in
          P}B_{\Omega}(u)}}=\abs{\hat V}^{O(1)}\exp{O(L)}
  \]
\end{proof}

\begin{proof}[Proof of Theorem \ref{thm:comp-marg}]
  Assume $\abs{V}=n$. Let $C$ and $\theta$ be the constants in
  Lemma \ref{lem:cut-exist}.  Fix some $L\ge C\log n$ and denote $S$
  the termination set of $\marg(\Omega,v,x,L)$.  Let
  $p_{\Omega,v,L}(x)=\marg(\Omega,v,x,L)$ and
  $T=\Tsaw{G[V\setminus\Lambda]}{v}$.  Then it follows from Lemma
  \ref{lem:err-uppbound} and Lemma \ref{lem:err-uppbound-color} that
  \[
    \abs{\log\tuple{p_{\Omega,v,L}(x)}-\log\tuple{\Pr[\Omega]{c(v)=x}}}\le\err_{T,S}.
  \]
  We also have
  \[
    \err_{T,S}\le n^{O(1)}\cdot\sum_{k=L+1}^{\theta
      L}\err_\delta(v,k)\le n^{C_1}\cdot\exp{-C_2\cdot L}
  \]
  for some universal constants $C_1,C_2>0$.

  Thus for some $L=O(\log n)$ and $L\ge C\log n$, it holds that
  \[
    1-O\tuple{\frac{1}{n^3}}\le \frac{\hat
      p_{\Omega,v,S_v}(x)}{\Pr[\Omega]{c(v)=x}} \le
    1+O\tuple{\frac{1}{n^3}}.
  \]
  The running time of the algorithm directly follows from Lemma
  \ref{lem:time-potts}.
\end{proof}

\subsection{The sampling algorithm}

\begin{theorem}
  Let $\mathcal{G}$ be a family of graphs satisfying the conditions in
  Theorem \ref{thm:comp-marg}. There exists an FPTAS to compute the
  partition function of Potts model with parameter $q$ and $\beta$ for
  every graph in $\mathcal{G}$.
\end{theorem}
\begin{proof}
  Let $\Omega=(G,\varnothing,\varnothing)$ be an instance of Potts
  model, where $G(V,E)\in\mathcal{G}$. Without loss of generality, we
  give an algorithm to compute an approximation of the partition
  function $\hat Z(\Omega)$ satisfying
  \[
    1-O\tuple{\frac{1}{n^2}}\le \frac{\hat Z(\Omega)}{Z(\Omega)}\le
    1+O\tuple{\frac{1}{n^2}}.
  \]
  Since our family of instances of Potts model is ``self-embeddable''
  in the sense of \cite{sinclair1989approximate}, the algorithm can be boosted
  into an FPTAS.
  
  Assume $V=\set{v_1,\dots,v_n}$. First find a configuration
  $\sigma\in [q]^V$ such that $w_G(\sigma)>0$.  This task is trivial
  when $\beta>0$. When $\beta=0$, since $G$ is $q$-colorable, we can
  also do it in polynomial time:
  \begin{itemize}
  \item If the graph is not empty, then choose a vertex $v$ and find a
    feasible coloring of $B(v)$. Then remove $B(v)$ from the graph and
    repeat the process.
  \end{itemize}
  If $G$ is $q$-colorable, then $G[V\setminus B(v)]$ is colorable as
  the boundary of $B(v)$ consists of low-degree vertices, thus the
  above process will end with a proper coloring of $G$, which is the
  union of colorings found at each step. The process terminates in
  polynomial time since $\mathcal{G}$ is locally sparse and thus the
  size of every $B(v)$ is $O(\log n)$.

  With $\sigma$ in hand, we have
  \begin{align*}
    Z(\Omega)=w_G(\sigma)/\Pr[\Omega]{c(V)=\sigma}
    &=w_G(\sigma)\tuple{\Pr[\Omega]{\bigwedge_{i=1}^nc(v_i)=\sigma(v_i)}}^{-1}\\
    &=w_G(\sigma)\tuple{\prod_{i=1}^n\Pr[\Omega]{c(v_i)=\sigma(v_i)\mid
      \bigwedge_{j=1}^{i-1}c(v_j)=\sigma(v_j)}}^{-1}
  \end{align*}
  For every $i\in[n]$, let $\Omega_i=(G,\Lambda_i,\sigma_i)$ where
  $\Lambda_i=\set{v_1,\dots,v_{i-1}}$ and $\sigma_i(v_j)=\sigma(v_j)$
  for every $j=1,\dots,i-1$. We have
  \[
    Z(\Omega)=w_G(\sigma)\tuple{\prod_{i=1}^n\Pr[\Omega_i]{c(v_i)=\sigma(v_i)}}^{-1}.
  \]
  Note that the graph class $\mathcal{G}$ is closed under the
  operation of fixing some vertex to a specific color, we can apply
  Theorem \ref{thm:comp-marg} for every $\Omega_i$ and obtain
  $\hat p_i$ such that
  \[
    1-O\tuple{\frac{1}{n^3}}\le \frac{\hat
      p_i}{\Pr[\Omega_i]{c(v_i)=\sigma(v_i)}}\le
    1+O\tuple{\frac{1}{n^3}}.
  \]
  Let $\hat Z(\Omega)=w_G(\sigma)\tuple{\prod_{i=1}^n\hat p_i}^{-1}$,
  then Theorem \ref{thm:comp-marg} implies that
  \[
    1-O\tuple{\frac{1}{n^2}}\le \frac{\hat Z(\Omega)}{Z(\Omega)}\le
    1+O\tuple{\frac{1}{n^2}}.
  \]
\end{proof}

Our approximate counting algorithm implies a sampling algorithm via
Jerrum-Valiant-Vazirani reduction\cite{samp_JVV86}.

\begin{corollary}\label{cor:fpaus}
  Let $q>2$ and $0\le \beta<1$ be two constants. For a family of
  graphs $\mathcal{G}$ satisfying conditions in Theorem
  \ref{thm:comp-marg}, and every graph $G(V,E)\in\mathcal{G}$ with
  $\abs{V}=n$, there exists an algorithm $\mathcal{S}$ such that for
  any $\epsilon>0$ with high probability $\mathcal{S}$ returns a
  random configuration in $[q]^{V(G)}$ from a distribution that is
  within total variation distance $\epsilon$ from the Gibbs
  distribution $\mu_{G}$ for the $q$-state Potts model with activity
  $\beta$. And the running time of $\mathcal{S}$ is in polynomial in
  $n$ and $\log\frac{1}{\epsilon}$. When $\beta=0$, i.e.~for $q>3d+4$,
  with high probability $\mathcal{S}$ is an FPAUS (fully
  polynomial-time almost uniform sampler) for proper $q$-colorings of
  $G$.
\end{corollary}

\section{Random Graphs}

In this section, we prove Theorem \ref{main-thm-random}. We first prove the following
properties of $\mathcal{G}(n,d/n)$.

\begin{theorem}\label{thm:gnp-property}
  Let $d$ be a sufficiently large constant, $q>3(1-\beta)+4$ and
  $G=(V,E)\sim\mathcal{G}(n,d/n)$. Then with probability $1-o(1)$, the
  following holds
  \begin{itemize}
  \item there exist two universal positive constants $C>0,\gamma<1$
    such that $\err_\delta(v,\ell)<n^C\gamma^\ell$ for all $v\in V$
    and for all $\ell=o(\sqrt{n})$;
  \item if $\beta=0$, then $G$ is $q$-colorable;
  \item there exists a universal constant $C>0$ such that for every
    path $P$ in $G$ of length $\ell$, $\abs{B(P)}\le C(\ell+\log n)$.
  \end{itemize}
\end{theorem}

Note that the first property in above theorem impose an upper bound on
$\ell$.  This is not harmful as our algorithms for FPTAS and sampling
only require the property holds for $\ell=O(\log n)$. Thus Theorem
\ref{thm:gnp-property} and Corollary \ref{cor:fpaus} together imply
Theorem \ref{main-thm-random}.




It is well-known that when $\beta=0$, $G$ is $q$-colorable with high
probability (see e.g.,~\cite{grimmett1975colouring}), we verify the
first property in Lemma \ref{lem:gnp-decay} and the third property in
Lemma \ref{lem:gnp-sparse}.

\subsection{Correlation decay in random graphs}

\begin{lemma}\label{lem:gnp-decay}
  Let $d>1$, $0\le\beta< 1$ and $q> 3(1-\beta)d+4$ be
  constants. Let
  $G(V,E)\sim\mathcal{G}(n,d/n)$. 
  There exist two positive constants $C>0$ and $\gamma<1$ such that
  with probability $1-O\tuple{\frac{1}{n}}$, for every $v\in V$ and
  every $\ell=o(\sqrt{n})$, it holds that
  \[
    \err_\delta(v,\ell)\le n^C\gamma^\ell
  \]

\end{lemma}

We first prove a technical lemma.

\begin{lemma}\label{lem:gnp-decay-tech}
  Let $0\le\beta<1$ be a constant. Let
  $f_q(d):\mathbb{R}^{\ge 0}\to\mathbb{R}^{\ge 0}$ be a piece wise function
  defined as
  \[
    f_q(d)\defeq
    \begin{cases}
      \frac{2(1-\beta)}{q-1-(1-\beta) d} & \mbox{if }d\le\frac{q-1}{1-\beta}-2\\
      1 & \mbox{otherwise.}
    \end{cases}
  \]
  Let $X$ be a random variable distributed according to binomial
  distribution $\mathrm{Bin}(n,\frac{\Delta}{n})$ where $\Delta>1$ is
  a constant. Then for $q\ge 3(1-\beta)\Delta+2$ and all sufficiently
  large $n$, it holds that $\E{f_q(X)}<\frac{1}{\Delta}$.
\end{lemma}

\begin{proof}
  Let $\lambda=1-\beta$. Since $f(d)$ is decreasing in $q$, we can
  assume $q=3\lambda\Delta+2$.  Note that
  \[
    \E[d\sim\mathrm{Bin}\tuple{n,\frac{\Delta}{n}}]{f(d)}\le\frac{1}{\Delta}
    \iff
    \E[d\sim\mathrm{Bin}\tuple{n,\frac{\Delta}{n}}]{1-f(d)}\ge\frac{\Delta-1}{\Delta}.
  \]
  Let $g(x)\defeq 1-f(x)$, then
  \[
    \E[d\sim\mathrm{Bin}\tuple{n,\frac{\Delta}{n}}]{1-f(d)}=\sum_{k=0}^{\lfloor
      \frac{q-1}{\lambda}-2\rfloor} g(k)\cdot p(k)
  \]
  where
  $p(k)=\binom{n}{k}\tuple{\frac{\Delta}{n}}^k\tuple{1-\frac{\Delta}{n}}^{n-k}$.

  Define
  \begin{align*}
    \tilde g(x)& \defeq 1 -\frac{2\lambda}{q-1-\lambda\Delta}
                 -\frac{2\lambda^2(x-\Delta)}{\tuple{q-1-\lambda\Delta}^2}
                 -\frac{2\lambda^3(x-\Delta)^2}{\tuple{q-1-\lambda\Delta}^3}
                 -\frac{2\lambda^4(x-\Delta)^3}{\tuple{q-1-\lambda\Delta}^4}\\
               &\quad\;-\frac{2\lambda^5(x-\Delta)^4}{\tuple{q-1-\lambda\Delta}^5}
                 -\frac{2\lambda^6(x-\Delta)^5}{\tuple{q-1-\lambda\Delta}^6}
                 -\frac{2\lambda^6(x-\Delta)^6}{\tuple{q-1-\lambda\Delta}^6}.
  \end{align*}
  Then
  \[
    g(x)-\tilde g(x)=
    \frac{2\lambda^6(q-1-\lambda-x\lambda)(x-\Delta)^6}{(q-1-x\lambda)(q-1-\lambda\Delta)^6},
  \]
  which is positive for $x\le\lfloor\frac{q-1}{\lambda}-2\rfloor$.

  We now prove that
  \[
    \sum_{k=0}^{\lfloor\frac{q-1}{\lambda}-2\rfloor}\tilde g(k)\cdot
    p(k)\ge \frac{\Delta-1}{\Delta}.
  \]
  The expectation of $\tilde g(k)$ can be computed directly:
  \[
    \E{\tilde
      g(k)}=\frac{1}{n^5(q-1-\lambda\Delta)^6}\cdot\tuple{C_5n^5+C_4n^4\pm
      O(n^3)},
  \]
  where
  \begin{align*}
    C_5&=1 - 2 \lambda + (12 \lambda - 20 \lambda^2 - 2 \lambda^3 - 2
         \lambda^4 - 2 \lambda^5 -
         4 \lambda^6) \Delta \\
       &\quad+ (60 \lambda^2 - 80 \lambda^3 - 12 \lambda^4 - 14 \lambda^5
         - 74 \lambda^6) \Delta^2 + (160 \lambda^3 - 160 \lambda^4 - 24
         \lambda^5 -
         50 \lambda^6) \Delta^3 \\
       &\quad + (240 \lambda^4 - 160 \lambda^5 - 16 \lambda^6) \Delta^4 +
         (192 \lambda^5 -
         64 \lambda^6) \Delta^5 + 64 \lambda^6 \Delta^6;\\
    C_4&=2 \lambda^3 (1 + 3 \lambda + 7 \lambda^2 + 46 \lambda^3)
         \Delta^2 + 2 \lambda^3 (6 \lambda + 18 \lambda^2 +
         234 \lambda^3) \Delta^3\\
       &\quad+ 2 \lambda^3 (12 \lambda^2 + 69 \lambda^3) \Delta^4 + 16
         \lambda^6 \Delta^5.
  \end{align*}

  Since $C_4>0$, thus for sufficiently large $n$, it holds that
  \[
    \E{\tilde g(x)}\ge \frac{C_5}{(q-1-\lambda\Delta)^6}.
  \]
  We also have that
  \[
    \E{\tilde
      g(x)}=\sum_{k=0}^{\lfloor\frac{q-1}{\lambda}-2\rfloor}\tilde
    g(k)\cdot
    p(k)+\sum_{k=\lfloor\frac{q-1}{\lambda}-1\rfloor}^{n}\tilde
    g(k)\cdot p(k)
  \]
  It can be verified that $\tilde g(x)$ is monotonically decreasing in
  $x$ when $x\ge \frac{q-1}{\lambda}-2$ and
  $\tilde
  g\tuple{\frac{q-1}{\lambda}-2}=-\tuple{\frac{1+2\lambda(\Delta-1)}{1+2\lambda\Delta}}^6<0$.
  
  Thus we have
  \[
    \sum_{k=0}^{\lfloor\frac{q-1}{\lambda}-2\rfloor}\tilde g(k)\cdot
    p(k) \ge \E{\tilde
      g(x)}\ge\frac{C_5}{(q-1-\lambda\Delta)^6}=\frac{\Delta-1}{\Delta}+h(\Delta)
  \]
  where
  \begin{align*}
    h(\Delta) &=\left(1 + 10 \lambda \Delta + (40 \lambda^2 - 2
                \lambda^3 - 2 \lambda^4 - 2 \lambda^5 - 4 \lambda^6)
                \Delta^2\right.\\
              &\quad\left.+ (80 \lambda^3 - 12 \lambda^4 - 14 \lambda^5 - 74
                \lambda^6)
                \Delta^3 + (80 \lambda^4 - 24 \lambda^5 - 50 \lambda^6) \Delta^4\right)\\
              &\quad\left.+ (32 \lambda^5 - 16 \lambda^6) \Delta^5\right)\cdot
                \tuple{\Delta(1+2\lambda\Delta)^6}^{-1}.
  \end{align*}
  It can be verified that $h(\Delta)$ is positive for every
  $0<\lambda<1$ and $\Delta\ge 1$.
\end{proof}

\begin{proof}[Proof of Lemma \ref{lem:gnp-decay}]
  Let $v\in V$ be arbitrary fixed and $T_v=\Tsaw{G}{v}$ and $\ell>0$
  be an integer. By linearity of expectation, we have
  \[
    \E{\err_\delta(v,\ell)}\le
    n^\ell\tuple{\frac{d}{n}}^\ell\E{\prod_{i=1}^\ell\delta(\deg[G]{v_i})\mid
      P=(v,v_1,\dots,v_\ell)\mbox{ is a path}}.
  \]
  Fix a tuple $P=(v,v_1,\dots,v_\ell)$. To calculate the expectation,
  we construct an independent sequence whose product dominates
  $\prod_{i=1}^\ell\delta(\deg[G]{v_i})$ as follows.

  Conditioning on $P=(v,v_1,\dots,v_\ell)$ being a path in $G$. Let
  $X_1,X_2,\dots,X_\ell$ be random variables such that each $X_i$
  represents the number of edges between $v_i$ and vertices in
  $V\setminus\set{v_1,\dots,v_\ell}$; and let $Y$ be a random variable
  representing the number of edges between vertices in
  $\set{v_1,\dots,v_\ell}$ except for the edges in the path
  $P=(v,v_1,\dots,v_\ell)$. Then $X_1,\dots,X_\ell,Y$ are mutually
  independent binomial random variables with each $X_i$ distributed
  according to $\mathrm{Bin}(n-\ell,\frac{d}{n})$ and $Y$ distributed
  according to $\mathrm{Bin}(\binom{\ell}{2}-\ell+1,\frac{d}{n})$, and
  for each $v_i$ in the path we have $\deg[G]{v_i}=X_i+2+Y_i$ with
  some $Y_1+Y_2+\dots+Y_\ell=2Y$.

  Note that $\delta(\deg[G]{v_i})=f_q(\deg[G]{v_i})$ where the
  function $f_q(x)$ is defined in Lemma \ref{lem:gnp-decay-tech}.
  Note that the ratio $f_q(x)/f_q(x-1)$ is always upper bounded by
  $2$, and we have $f_q(x+1)\le f_{q-1}(x)$. Thus, conditioning on
  that $P=(v,v_1,\dots,v_\ell)$ is a path, the product
  $\prod_{i=1}^\ell\delta_{q,\beta}(\deg[G]{v_i})$ can be bounded as
  follows:
  \[
    \prod_{i=1}^\ell\delta(\deg[G]{v_i})=\prod_{i=1}^\ell
    f_q(X_i+Y_i+2)\le 2^{2Y}\prod_{i=1}^\ell f_{q-2}(X_i).
  \]
  Let $d'=\frac{q-4}{3(1-\beta)}$, then we have $d'>d$. Let $X$ be a
  binomial random variable distributed according to
  $\mathrm{Bin}(n,\frac{d'}{n})$, thus $X$ probabilistically dominates
  every $X_i$ whose distribution is
  $\mathrm{Bin}(n-\ell,\frac{d}{n})$. Since $X_1,X_2,\dots,X_\ell,Y$
  are mutually independent conditioning on $P=(v,v_1,\dots,v_\ell)$
  being a path in $G$, for any $P=(v,v_1,\dots,v_\ell)$ we have
  \[
    \E{\prod_{i=1}^\ell\delta(\deg[G]{v_i})\mid P\mbox{ is a path}}\le
    \E{4^Y\prod_{i=1}^\ell f_{q-2}(X_i)} \le
    \E{4^Y}\E{f_{q-2}(X)}^\ell.
  \]
  Recall that
  $Y\sim\mathrm{Bin}\tuple{\binom{\ell}{2}-\ell+1,\frac{d}{n}}$, the
  expectation $\E{4^Y}$ can be bounded as
  \[
    \E{4^Y}\le\sum_{k=0}^{\ell^2}4^k\binom{\ell^2}{k}\tuple{\frac{d}{n}}^k\tuple{1-\frac{d}{n}}^{\ell^2-k}
    =\tuple{1+\frac{3d}{n}}^{\ell^2}\le\exp{\frac{3d\ell^2}{n}}.
  \]
  Since $q-2\ge 3(1-\beta)d'+2$, it follows from Lemma
  \ref{lem:gnp-decay-tech} that
  $\E{f_{q-2}(X)}\le\frac{1}{d'}=\frac{3(1-\beta)}{q-4}$. Therefore,
  \[
    \E{\prod_{i=1}^\ell\delta(\deg[G]{v_i})\mid P\mbox{ is a path}}
    \le\exp{\frac{3d\ell^2}{n}}\tuple{\frac{3(1-\beta)}{q-4}}^{\ell}
    \le\frac{1}{d^\ell}\cdot\exp{-\ell\log\tuple{\frac{q-4}{3d(1-\beta)}}+\frac{3d\ell^2}{n}}.
  \]
  Since $\ell=o(\sqrt{n})$,
  \[
    \E{\err_\delta(v,\ell)}\le
    \exp{-\ell\log\tuple{\frac{q-4}{3d(1-\beta)}}+o(1)}.
  \]
  Then the lemma follows from the Markov inequality and the union
  bound.
\end{proof} \subsection{Locally sparse for random graphs}

\begin{lemma}\label{lem:gnp-sparse}
  Let $\eps>0$ be some fixed constant. Let $d$ be a sufficiently large
  number, $q\ge (2+\eps)d$ and $0\le \beta<1$ be constants. Let
  $G=(V,E)\sim\mathcal{G}(n,d/n)$. There exists a constant $C>0$
  such that with probability $1-O\tuple{\frac{1}{n}}$, for every path
  $P$ in $G$ of length $\ell$, $\abs{B(P)}\le C(\ell+\log n)$.

\end{lemma}

Given $P=(v_1,\dots,v_L)$, we are going to upper bound the
probability
\begin{equation}
  \label{eq:pb}
  \Pr{\abs{B(P)}\ge t\mid\mbox{$P$ is a path}}
\end{equation}
for every $t>0$.

A vertex $v$ is a high-degree vertex if
$\deg[G]{v}\ge \frac{q-1}{1-\beta}-2$.  Thus the probability
\eqref{eq:pb} is maximized when $\beta=0$. Note that conditioning on
$P$ is a path gives each vertex at most two degrees, we can redefine
the notion of ``high-degree'' as $\deg[G]{v}\ge q-5$ and drop the
condition that $P$ is a path. Thus it is sufficient to upper bound
\[
  \Pr{\abs{B(P)}\ge t}
\]
with our new definition of high-degree vertices.

Let $G=(V,E)$ be a graph. We now describe a BFS procedure to generate
$B^*(P)\defeq B(P)\cup\partial B(P)$. Since $B^*(P)$ is always a
superset of $B(P)$, it is sufficient to bound
$\Pr{\abs{B^*(P)}\ge t}$. For a vertex $v\in V$, we use $N_G(v)$ to
denote the set of neighbors of $v$ in $G$.

Initially, we have a counter $i=0$, a graph $G_0=G$, a set of active
vertices $\mathcal{A}_0=\set{v_1,v_2,\dots,v_L}$ and a set of used vertices
$\mathcal{U}_0=\varnothing$.

\paragraph{(P1)}
\begin{enumerate}
\item Increase the counter $i$ by one.
\item (If $i\le L$) Define
  $G_i(V_i,E_i)=G_{i-1}[V_{i-1}\setminus\set{v_i}]$. Let
  $\mathcal{U}_i=\mathcal{U}_{i-1}\cup \set{v_i}$.  Let
  $\mathcal{A}_i=(\mathcal{A}_{i-1}\cup
  N_{G_{i-1}}(v_i))\setminus\mathcal{U}_i$.  Goto 1.
\item (If $i> L$) Terminate if
  $\mathcal{A}_{i-1}=\varnothing$. Otherwise, let
  $u\in \mathcal{A}_{i-1}$ and let
  $\mathcal{U}_i=\mathcal{U}_{i-1}\cup \set{u}$.
  \begin{enumerate}
  \item (If $\abs{N_{G}(u)}\ge q-5$) Define
    $G_i(V_i,E_i)=G_{i-1}[V_{i-1}\setminus\set{u}]$. Let
    $\mathcal{A}_i=(\mathcal{A}_{i-1}\cup
    N_{G_{i-1}}(v_i))\setminus\mathcal{U}_i$. Goto 1.
  \item (If $\abs{N_{G}(u)}< q-5$) Define $G_i=G_{i-1}$. Let
    $\mathcal{A}_i=\mathcal{A}_{i-1}\setminus\mathcal{U}_i$.  Goto 1
  \end{enumerate}
\end{enumerate}

The following proposition is immediate: 

\begin{proposition}
  Assume the algorithm terminates at step $t$, then
  $B^*(P)=\mathcal{U}_{t-1}$ and $\abs{B^*(P)}=t-1$
\end{proposition}

Let $R=\set{r_1,r_2,\dots,r_L}$ be a set and each $r_i$ is the root of
tree $T_i$. We now describe a BFS procedure to explore these $L$
trees.  For a vertex $v$, we use $C(v)$ to denote its children.

Initially, we have a counter $i=0$ and a set of active vertices
$\mathcal{B}_0=R$.

\paragraph{(P2)}
\begin{enumerate}
\item Increase the counter $i$ by one.
\item (If $i\le L$) Let
  $\mathcal{B}_i=(\mathcal{B}_{i-1}\cup
  C(r_i))\setminus\set{r_i}$. Goto 1.
\item (If $i>L$) Terminate if
  $\mathcal{B}_{i-1}=\varnothing$. Otherwise, let
  $w\in \mathcal{B}_{i-1}$
  \begin{enumerate}
  \item (If $\abs{C(w)}\ge \frac{q-5}{2}$) Let
    $\mathcal{B}_i=(\mathcal{B}_{i-1}\cup C(u))\setminus\set{w}$. Goto
    1.
  \item (If $\abs{C(w)}<\frac{q-5}{2}$) Let
    $\mathcal{B}_i=\mathcal{B}_{i-1}\setminus\set{w}$. Goto 1.
  \end{enumerate}
\end{enumerate}

Now assume $G\sim\mathcal{G}(n,d/n)$ and for every $i\in[L]$, $T_i$ is
a branching process with distribution $\mathrm{Bin}(n,d/n)$, i.e.,
each $C(u)\sim\mathrm{Bin}([n],d/n)$. We can implement the (P1) when
at each step $i$, the vertex $u$ chosen from the active set sample its
neighbors $N_{G_{i-1}}(u)$ according to $\mathrm{Bin}(V_i,d/n)$. This
random process can be coupled with $\mathcal{G}(n,d/n)$ such that
$B^*(P)$ found by it is always a superset of the one in
$\mathcal{G}(n,d/n)$.

We now construct a coupling of (P1) and (P2) with the property that
the later one always terminates no earlier than the former one.

At each step $i\ge 1$, let $u$ and $w$ be the vertex chosen from
$\mathcal{A}_i$ and $\mathcal{B}_i$ respectively ($u=v_i$ and $w=r_i$
if $i\le L$).  Then
$\abs{N_{G_{i-1}}(u)}\sim \mathrm{Bin}(\abs{V_i},d/n)$. We couple it
with some $x\sim\mathrm{Bin}(n,d/n)$ with the property that
$x\ge \abs{N_{G_{i-1}}(v_i)}$ and let $C(w)$ be a set with $x$
elements.

\begin{lemma}
  For every $i\ge 0$, the following two properties hold:
  \begin{enumerate}
  \item[(i1)] There exists a surjective mapping $F_i$ from
    $\mathcal{B}_i$ to $\mathcal{A}_i$ in each step $i$.
  \item[(i2)] For every $u\in\mathcal{A}_i$, we use $n_i(u)$ to denote
    the number of $w\in\mathcal{B}_i$ such that $F_i(w)=u$. Then for
    every $u\in\mathcal{A}_i$,
    $n_i(u)\ge \abs{N_G(u)}-\abs{N_{G_{i}}(u)}$.
  \end{enumerate}
\end{lemma}
\begin{proof}
  We apply induction on $i$ to prove the lemma.

  When $i=0$, we let $F_0:\mathcal{B}_0\to\mathcal{A}_0$ be the
  function that $F_0(r_j)=v_j$ for every $j\in[L]$. Then both
  properties hold trivially.

  Assume the lemma holds for smaller $i$. If $i\le L$, since by our
  coupling, $\abs{C(r_i)}\ge \abs{N_{G_{i-1}}(v_i)}$, we can construct
  $F_i$ by extending $F_{i-1}$ with an arbitrary surjective mapping
  from $C(r_i)$ to $N_{G_{i-1}}(v_i)$. For every $u'\in\mathcal{A}_i$,
  if $ u'\in N_{G_{i-1}}(v_i)$, then $n_i(u')\ge n_{i-1}(u')+1$ and
  $\abs{N_{G_{i-1}(u)}} -\abs{N_{G_{i}}(u')}=1$; otherwise
  $n_i(u')=n_{i-1}(u')$ and
  $\abs{N_{G_{i-1}}(u)}=\abs{N_{G_{i}}(u')}$. Induction hypothesis
  implies both $(i1)$ and $(i2)$ hold.

  If $i>L$, we have to distinguish between cases:
  \begin{itemize}
  \item (If $\abs{N_G(u)}\ge q-5$ and
    $N_{G_{i-1}}(u)\ge \frac{q-5}{2}$) We construct $F_i$ by extending
    $F_{i-1}$ with an arbitrary surjective mapping from $C(w)$ to
    $N_{G_{i-1}}(u)$, the same argument as $i\le L$ case proves $(i1)$
    and $(i2)$.
  \item (If $\abs{N_G(u)}\ge q-5$ and $N_{G_{i-1}}(u)<\frac{q-5}{2}$)
    In this case, by induction hypothesis, we know that
    \[
      n_{i-1}(u)\ge \abs{N_G(u)}-\abs{N_{G_{i-1}}(u)}\ge
      \frac{q-5}{2}>N_{G_{i-1}}(u).
    \]
    Choose a surjective $f$ from $F_{i-1}^{-1}(u)$ to $N_{G_{i-1}}(u)$
    and construct $F_i$ from $F_{i-1}$ by replacing the mapping on
    $F_{i-1}^{-1}(u)$ by $f$. This is safe since
    $u\not\in\mathcal{A}_i$. The same argument as before proves $(i1)$
    and $(i2)$.
  \item (If $\abs{N_G(u)}<q-5$) Construct $F_i=F_{i-1}$. Since
    everything does not change, the induction hypothesis implies
    $(i1)$ and $(i2)$.
  \end{itemize}
\end{proof}

The first property above guarantees that (P2) terminates no earlier
than (P1) and thus its stopping time is an upper bound for the size of
$B^*(P)$ found by (P1).

(P2) can be modeled as follows:
\begin{enumerate}
\item Let $X\sim\mathrm{Bin}(n,d/n)$ and $X_1,X_2\dots$ be an infinite
  sequence of independent random variables defined as follows
  \begin{itemize}
  \item For $i=1,2,\dots,L$, $X_i$ is an independent copy of $X$;
  \item For $i>L$, $X_i$ has following distribution
    \[
      X_i=
      \begin{cases}
        0 & \mbox{if $X<(q-5)/2$}\\
        X & \mbox{otherwise}.
      \end{cases}
    \]
  \end{itemize}
\item $Y_1,Y_2,\dots$ is an infinite sequence of random variables that
  $Y_0=L$ and $Y_i=Y_{i-1}+X_i-1$ for every $i\ge 1$.
\item $Z=\min_t\set{Y_t=0}$.
\end{enumerate}

The above process is identical to (P2), thus we have
\begin{proposition}
  (P2) terminates after step $t$ if and only if $Z>t$.
\end{proposition}

Note that $Z>t$ implies $Y_t\ge 0$, we turn to bound the latter.

\begin{lemma}\label{lem:gnp-chernoff}
  There exist two constants $C_1,C_2>0$ depending on $d$ and $\eps$
  such that
  \[
    \Pr{Y_t\ge 0}\le \exp{-C_1t+C_2L}.
  \]
\end{lemma}
\begin{proof}
  By the definition,
  $Y_{t+L}=L-(t+L)+\sum_{i=1}^{L+t}
  X_i=-t+\sum_{i=1}^LX_i+\sum_{i=L+1}^{L+t}X_i$.
  We know the distribution of $X_i$s and we now compute their moment
  generating function.  For every $s>0$, it holds that
  \[
    \Pr{Y_{t+L}\ge 0} =\Pr{e^{sY_{t+L}}\ge 1} \le
    \E{e^{sY_{t+L}}}=e^{-st}\tuple{\E{e^{sX}}}^L\tuple{\E{e^{sX_{L+1}}}}^t.
  \]
  Recall that $X\sim\mathrm{Bin}(n,d/n)$, we have
  $\E{e^{sX}}=\tuple{1+\frac{d}{n}(e^s-1)}^n\le e^{d(e^s-1)}$. Let
  $p=(q-5)/2$, we have
  \begin{align*}
    \E{e^{sX_{L+1}}}
    &=\Pr{X< p}+\sum_{k=\lfloor p\rfloor}^n e^{sk}\cdot\Pr{X=k}\\
    &\le 1+\sum_{k=\lfloor p\rfloor}^n e^{sk}\cdot\Pr{X\ge k}\\
    &\le\exp{\sum_{k=\lfloor p\rfloor}^\infty e^{sk}\cdot\Pr{X\ge k}}
  \end{align*}
  By Chernoff bound, for sufficiently large $d$, we have for some
  choices of $s>0$ and $C_1>0$,
  \[
    \sum_{k=\lfloor p\rfloor}^\infty e^{sk}\cdot\Pr{X\ge k}-s<-C_1'.
  \]
  Let $C_2'=d(e^s-1)$, we have
  \[
    \Pr{Y_{t+L}\ge 0} \le \exp{-C_1't+C_2'L}.
  \]
  This implies for some constants $C_1,C_2>0$,
  \[
    \Pr{Y_t\ge 0} \le \exp{-C_1t+C_2L}.
  \]
  
\end{proof}

\begin{proof}[Proof of Lemma \ref{lem:gnp-sparse}]
  By Lemma \ref{lem:gnp-chernoff} and the union bound, the probability
  that there exists a path $P$ in $G$ of length $\ell$ such that
  $\abs{B(P)}\ge t $ is upper bounded by
  \begin{align*}
    n\cdot n^\ell\cdot\tuple{\frac{d}{n}}^{\ell}\cdot\Pr{Y_t\ge 0}
    &\le n\cdot d^{\ell}\cdot\exp{-C_1t+C_2\ell}=O\tuple{\frac{1}{n}}
  \end{align*}
  for $t=C(\ell+\log n)$ and sufficiently large constant $C$.
\end{proof}

\bibliographystyle{alpha} \bibliography{paper}

\end{document}